\documentclass[a4paper,onecolumn,11pt,accepted=2025-09-22]{quantumarticle}
\pdfoutput=1
\usepackage[utf8]{inputenc}
\usepackage[english]{babel}
\usepackage[T1]{fontenc}
\usepackage{graphicx, caption}
\usepackage{epstopdf}
\usepackage{amsfonts}
\usepackage{bm}
\usepackage{cite}
\usepackage{amssymb}
\usepackage{amsthm}
\usepackage{physics}
\usepackage{color}
\usepackage{comment}
\usepackage{mathtools}
\usepackage{braket}
\usepackage{enumitem}
\usepackage{hyperref}
\usepackage{balance}
\usepackage{flushend}
\usepackage{mathrsfs}
\usepackage{setspace}
\usepackage{amsmath}
\usepackage{cleveref}
\usepackage{adjustbox}
\usepackage{lscape}
\usepackage{tikz}
\usepackage{yquant}
\usepackage{pgfplots}
\usepackage{subcaption}
\usepackage{algorithmicx}
\usepackage{algorithm}
\usepackage[noend]{algpseudocode}
\usepackage{array}
\usepackage{ragged2e}
\usepackage{authblk}

\algnewcommand\algorithmiclet{\textbf{Let}}
\algnewcommand\Let{\item[\algorithmiclet]}
\algnewcommand{\LineComment}[1]{\State \(\triangleright\) #1}

\pgfplotsset{compat = newest}

\newcommand{\C}{\mathbb{C}}
\newcommand{\E}{\mathbb{E}}

\newcommand{\normalLap}{\tilde{\Delta}_k}

\newcommand{\ibmOneName}{\textbf{QBNE-Chebyshev}}
\newcommand{\apersOneName}{\textbf{CBNE-Power}}
\newcommand{\apersTwoName}{\textbf{CBNE-Chebyshev}}
\newcommand{\thisPaperName}{\textbf{QBNE-Power}}
\newcommand{\sampleaxis}{\hspace*{-0.22cm}\resizebox{0.18\textwidth}{!}{\begin{tikzpicture}[every node/.style={font={\tiny}}]
  \draw[-] (-1.0,0)--(1.0,0);
  \foreach \i in {-1,-0.5,...,1.1} \draw (\i,0)--(\i,.1);
  \node[below] at (-0.92,0) {$10^2$};
  \node[below] at (-0.5,0) {$10^3$};
  \node[below] at (0.0,0) {$10^4$};
  \node[below] at (0.5,0) {$10^5$};
  \node[below] at (0.92,0) {$10^6$};
  \node[below] at (0,-0.3) {Sample};
\end{tikzpicture}}}

\newtheorem{theorem}{Theorem}[section]

\newtheorem{lemma}[theorem]{Lemma}
\theoremstyle{definition}
\newtheorem*{definition}{Definition}

\newtheorem*{remark}{Remark}

\title{Comparing quantum and classical Monte Carlo algorithms for estimating Betti numbers of clique complexes}

\author[1]{Ismail Yunus Akhalwaya\textsuperscript{*}}
\author[1]{Ahmed Bhayat}
\author[1]{Adam Connolly\textsuperscript{*}}
\author[1]{Steven Herbert}
\author[2]{Lior Horesh}
\author[3]{Julien Sorci\textsuperscript{*}}
\email{julien.sorci@quantinuum.com}
\homepage{*Equal Contribution}
\author[2]{Shashanka Ubaru}
\affil[1]{\textit{Quantinuum, Terrington House,
13-15 Hills Road, Cambridge CB2 1NL, United Kingdom}}
\affil[2]{\textit{IBM Research, USA}}
\affil[3]{\textit{Quantinuum, 1300 N 17th Street, Arlington, VA 22209
USA}}

\begin{document}
\maketitle

\begin{abstract}
Several quantum and classical Monte Carlo algorithms for Betti Number Estimation (BNE) on clique complexes have recently been proposed, though it is unclear how their performances compare. We review these algorithms, emphasising their common Monte Carlo structure within a new modular framework. We derive upper bounds for the number of samples needed to reach a given level of precision, and use them to compare these algorithms. By recombining the different modules, we create a new quantum algorithm with an exponentially-improved dependence in the sample complexity. We run classical simulations to verify convergence within the theoretical bounds and observe the predicted exponential separation, even though empirical convergence occurs substantially earlier than the conservative theoretical bounds.
\end{abstract}

\section{Introduction}

Given a graph, the clique complex is a geometric object that captures its clique information. An important topological invariant of a clique complex are the Betti numbers, $\beta_k$, which quantify the number of holes in the complex in a given dimension $k$. The Betti numbers have a long history in computational algebraic topology and data analysis: For fixed $k$, polynomial-time algorithms for Betti number estimation date back to the 1970s~\cite{Kannan1979}, and many recent applications of this problem have been found in the field of Topological Data Analysis~\cite{Carlsson2009,Amezquita2020, Skaf2022}. Despite this, it has been recently shown that, given as input a graph $G$ and dimension $k$, deciding if $\beta_k = 0$ in the clique complex of $G$ is $\texttt{QMA}_{1}\texttt{-Hard}$, and, when $G$ is clique-dense this decision problem is in \texttt{QMA}~\cite{Crichigno2022}. This means that, under widely-held computational assumptions, there is no efficient classical or quantum algorithm for computing the exact Betti numbers in all dimensions. 

This intractability of computing exact Betti numbers in arbitrary dimensions still leaves open the possibility of efficiently approximating $\beta_k$ divided by the number of $k$-simplices. This quantity is called the \emph{normalised} Betti number, and estimating it is the problem we consider here and refer to as BNE. BNE has a relatively short history compared to computing exact Betti numbers, appearing first in the literature of property testing for graphs~\cite{Elek2010}.  Fortunately, BNE is known to be in \texttt{BQP}~\cite{Cade2021} for general simplicial complexes that are efficiently sampleable, of which efficiently sampleable clique complexes are a special case. In favour of the power of quantum algorithms over classical, the same paper shows BNE to be classically intractable (\texttt{DQC1-Hard}) for general complexes~\cite{Cade2021} and leaves open the question of whether BNE remains classically intractable for dense clique complexes (the near term complexes of interest). Nevertheless, this \texttt{DQC1-Hardness} and the previously mentioned $\texttt{QMA}_{1}$-completeness of exact Betti number calculation of dense clique complexes provide strong evidence for the classical intractability of BNE for dense clique complexes.

In the absence of a definitive complexity result for BNE on clique complexes, progress has been made by designing classical and quantum algorithms for BNE with steadily improving asymptotic behaviour. The first \texttt{BQP} result for dense clique complexes was proved by Lloyd, Garnerone and Zanardi~\cite{Lloyd2016} by introducing a new polynomial time quantum algorithm based on quantum phase estimation. Ever since this result, new quantum~\cite{Hayakawa2022, McArdle, Gyurik2022, Akhalwaya2022} and classical~\cite{Apers2023, Berry2024} algorithms have been introduced, and it is still believed that there is a regime where quantum algorithms attain a super-polynomial advantage over classical algorithms for this problem~\cite{Berry2024}.

This paper studies quantum and classical algorithms for the BNE problem which share a similar Monte Carlo structure and take as input (classically provided) clique samples. This forms a natural sub-class of classical and quantum algorithms for the BNE problem. We choose to restrict our attention to this sub-class because we are interested in closely examining the \emph{exponential} speed-ups and sample costs of the different quantum Monte Carlo counter-parts to the classical random walk, ignoring any polynomial speed-ups of other quantum algorithms such as ~\cite{Lloyd2016, Berry2024} which process superpositions of cliques using non-random walk techniques.


The algorithms we consider produce an estimate by taking a matrix $M$ related to the combinatorial Laplacian of the clique complex, choosing a polynomial $p$ such that the trace of $p(M)$ is close to the normalised Betti number, and then performing a stochastic trace estimation of $p(M)$. We directly compare the Monte Carlo quantum algorithm of Akhalwaya et al.~\cite{Akhalwaya2022, akhalwaya2024topological} and the classical algorithms of Apers et al.~\cite{Apers2023}. The theoretical sample complexity bounds of these algorithms are presented in Table \ref{tab:comparison}. We note that the sample and query complexities presented in Table \ref{tab:comparison} do not have a clear dependence on the dimension $k$ or number of vertices in the complex $n$; These parameters will generally feature in the time complexity of performing one query in the algorithm, meaning either a step in the Markov chain for the classical algorithms, or a use of the relevant block-encoding for the quantum algorithms. More notably, the sample count for the quantum algorithm grows exponentially in $1/\sqrt{\delta}$, where $\delta$ is the spectral gap of the normalised Laplacian of the complex, and yields at most a polynomial advantage over the best classical algorithm, which we prove in Section~\ref{sec:nisq-tda}. We then introduce a new quantum algorithm which avoids this exponential sample count growth.

\begin{table}
    \centering
    \resizebox{\textwidth}{!}{%
    \begin{tabular}{|l|c|c|c|}
    \hline
        \textbf{Algorithm} & \textbf{Polynomial degree} $d$ & \textbf{Sample complexity} &\textbf{Query complexity}
        \\ \hline
        & & & \\[-2ex]
       \ibmOneName{} (Alg. \ref{alg:cap}) & $\frac{\log(1/\epsilon)}{\sqrt{\delta}}$ &
       $\mathcal{O}\left(d^2 \times 2^{10d}\right)$ & 
       $\mathcal{O}\left(d^3 \times 2^{10d}\right)$
       \\[1ex] \hline
       & & & \\[-2ex]
       \apersOneName{} (Alg. \ref{alg:apers_1}) & $\frac{\log(1/\epsilon)}{\delta}$ &
        $\mathcal{O}\left(\frac{1}{\epsilon^2} \times 2^{2d} \right)$  &
        $\mathcal{O}\left(\frac{1 }{\epsilon^2} \times d \times 2^{2d} \right)$
       \\[1ex] \hline
       & & & \\[-2ex]
       \apersTwoName{} (Alg. \ref{alg:apers_2}) & $\frac{\log(1/\epsilon)}{\sqrt{\delta}}$ & 
        $\mathcal{O}\left(d^3 \times 2^{8d}\right)$ &  $\mathcal{O}\left(d^4 \times 2^{8d}\right)$
       \\[1ex] \hline
       & & & \\[-2ex]
       \thisPaperName{} (Alg. \ref{alg:this_paper}) & $\frac{\log(1/\epsilon)}{\delta}$ &
        $\mathcal{O}\left(\frac{1}{\epsilon^2}\right)$  &
        $\mathcal{O}\left(\frac{1}{\epsilon^2} \times d\right)$
       \\[1ex] \hline 
    \end{tabular}}
    \caption{Comparison of the sample and query complexities of the four quantum and classical algorithms for estimating Betti numbers considered in this paper, where $\epsilon$ is the additive precision and $\delta$ is the spectral gap of the normalised combinatorial Laplacian of the complex. The sample and query complexities are given in terms of the polynomial degree $d$ given in the second column. We assume that the failure probability $\eta$ is constant. For more precise complexity formulas see Theorems~\ref{thm:complexity_ibm_one}, \ref{thm:complexity_apers1}, \ref{thm:complexity_apers}, \ref{thm:complexity_this_paper}, respectively.}
    \label{tab:comparison}
\end{table}

In Sections \ref{sec:top_background} and \ref{sec:trace_estimation}, we describe the background in topology and stochastic trace estimation which is necessary for the algorithms we consider. In Section \ref{sec:nisq-tda}, we review the Monte Carlo quantum algorithm of Akhalwaya et al.~\cite{Akhalwaya2022,Ubaru2021} which we refer to as \ibmOneName{} and present the revised complexity analysis mentioned in~\cite{akhalwaya2024topological}. Additionally, we prove that the sample count of this algorithm is exponential in $1/\sqrt{\delta}$. In Section \ref{sec:apers}, we review the two classical algorithms for normalised Betti number estimation introduced by Apers, Gribling, Sen and Szab\'{o}~\cite{Apers2023}, which we call \apersOneName{} and \apersTwoName{} and compare the complexities to \ibmOneName{}. In Section \ref{sec:power_method}, we recombine aspects of \apersOneName{} and \ibmOneName{} into a new quantum algorithm for BNE. We show that this results in a quantum algorithm which avoids the exponential dependence on $1/\delta$ present in \ibmOneName{} \footnote{Akhalwaya et al. in~\cite{akhalwaya2024topological} introduce a different quantum algorithm using qubitization which lies outside our lower-coherence Monte Carlo comparison framework.}. In Section \ref{sec:num_exp}, we simulate the algorithms on several small benchmark graphs and present both theoretical upper bounds and empirically observed sample counts for the minimum number of samples required for convergence.

\section{Simplicial complexes, Laplacians and Betti numbers}\label{sec:top_background}

 A \textit{simplicial complex} on a set $\{x_1,x_2,...,x_n\}$ is a collection of subsets $\Gamma$ of $\{x_1,x_2,...,x_n\}$ which is closed under subsets, meaning that if $\sigma$ is an element of $\Gamma$ and $\sigma^\prime$ is a subset of $\sigma$ then $\sigma^\prime$ is also in $\Gamma$. The elements of $\Gamma$ are called \textit{simplices} and a simplex is called a $k$-\textit{simplex} when its cardinality is $k+1$. Given a simplicial complex $\Gamma$ we write $S_k$ to denote the set of $k$-simplices in $\Gamma$. A simplicial complex of particular interest here is the \textit{clique complex} of a graph, which given a graph $G$ has a subset of the vertex set as a simplex if the vertices form a clique in $G$. We consider the vector space $\C \Gamma$ with the standard basis labelled by the elements of $\Gamma$. Furthermore, we define the $k^{th}$ \textit{boundary map} as the mapping $\partial_k : \C S_k \rightarrow \C S_{k-1}$ sending a $k$-simplex $\ket{\{x_{i_0},...,x_{i_k}\}}$ with $i_0 < i_1 < ... < i_k$ to
\begin{equation*}
    \partial_k \ket{\{x_{i_0},...,x_{i_k}\}} = \sum_{j=0}^k (-1)^{j+1} \ket{\{x_{i_0},...,x_{i_k}\} \setminus \{x_{i_j}\}},
\end{equation*}
and define the \textit{unrestricted boundary map} by $\partial = \oplus_{k=1}^n \partial_k$. A standard lemma in algebraic topology \cite[Lemma 2.1]{Hatcher2002} shows that $\partial_k \circ \partial_{k+1} = 0$, meaning that $\textnormal{Im}(\partial_{k+1}) \subseteq \ker(\partial_k)$, and therefore the quotient $\ker(\partial_k ) / \textnormal{Im}(\partial_{k+1})$ is well-defined, and called the $k^{th}$ \textit{homology group}. The $k^{th}$ \textit{Betti number} $\beta_k$ is then defined to be the dimension of this vector space, that is
\begin{equation*}
    \beta_k = \dim\Big( \ker(\partial_k ) / \textnormal{Im}(\partial_{k+1})\Big),
\end{equation*}
and is a quantitative expression for the number of $k$-dimensional holes in $\Gamma$. Similarly, the \textit{normalised Betti number} is defined as $\beta_k / |S_k|$. The $k^{th}$ \textit{combinatorial Laplacian} is defined as the mapping $\Delta_k : \C S_k \rightarrow \C S_k$
\begin{equation}\label{eqn:kLaplacian_boundary}
    \Delta_k = \partial_k^\dagger \partial_k + \partial_{k+1} \partial_{k+1}^\dagger. 
\end{equation}
The Combinatorial Hodge Theorem~\cite{Yu1983} shows that $\beta_k = \dim(\ker(\Delta_k))$, and therefore estimating the Betti number is the linear algebra task of calculating the nullity of the Laplacian \cite{Lim2020}. 
From the definition given in (\ref{eqn:kLaplacian_boundary}) we can see that $\Delta_k$ is positive semi-definite by noting that $\Delta_k = ( \partial_{k} + \partial^{\dagger}_{k+1})^{\dagger}( \partial_{k} + \partial^{\dagger}_{k+1})$. By the Laplacian Matrix Theorem of \cite[Theorem 3.4.4]{Goldberg2002} the diagonal elements of $\Delta_k$ are bounded by $n$, the size of the simplicial complex. This means that the eigenvalues of $\Delta_k$ are strictly in the range $[0,n]$.  Throughout the paper, we consider the normalised Laplacian $\tilde{\Delta}_k := \frac{1}{n}\Delta_k$ and its reflection $I - \tilde{\Delta}_k$, both of which have eigenvalues in $[0,1]$.
The nullity of $\tilde{\Delta}_k$ is equal to the nullity of $\Delta_k$ and the dimension of the $1$-eigenspace of $I-\tilde{\Delta}_k$. We assume throughout that $\delta > 0$ is a lower bound for the smallest positive eigenvalue of $\tilde{\Delta}_k$. These facts are used by the algorithms presented in this paper to estimate $\beta_k$.

In this paper we are interested in computing an additive estimate of the normalised Betti numbers, according to the following definition.
\begin{definition}
    Let $\epsilon, \eta >0$. We say that the estimator $\hat{\beta}_k$ is an $(\epsilon, \eta)$-\textit{estimator} if $\Pr\big[ \big|\hat{\beta}_k - \frac{\beta_k}{|S_k|}\big| \geq \epsilon \big] \leq \eta$.
\end{definition}
To map a simplicial complex onto the computational basis states of $(\C^2)^{\otimes n}$ we associate a $k$-simplex $\sigma$ in $\Gamma$ with the computational basis state $\ket{a_1,...,a_n} \in (\C^2)^{\otimes n}$ where $a_i = 1$ if $x_i \in \sigma$ and $a_i=0$ otherwise. Other more compact mappings from simplices to qubits are explored, for example, by McArdle et al.~\cite{McArdle} who also provide circuit constructions for the combinatorial Laplacian. These circuits are deeper than those presented in this paper.

\section{Stochastic trace estimation}
\label{sec:trace_estimation}

Let $M$ be an $N\times N$ matrix. The normalised trace of $M$ is $\frac{1}{N}\tr(M)$, which is the same as the \emph{average eigenvalue} of $M$. Each algorithm for Betti number estimation in this paper relies on a framework for estimating normalised traces called \textit{stochastic trace estimation}. This is typically applied to matrices $M$ which are too large to store directly but have efficient procedures for computing matrix-vector products such as $\bra{x}M\ket{x}$ for a vector $\ket{x}$. Rather than compute each of the diagonal entries of $M$, we define a random variable $X_M$ with expectation $\frac{1}{N}\tr(M)$ so that sampling from $X_M$ and averaging these samples gives an estimate for $\frac{1}{N}\tr(M)$. 
A typical example is to define $X_M = \bra{x} M \ket{x}$ where $\ket{x}$ is a vector chosen uniformly at random from the standard basis of $\C^N$.  
To estimate the number of samples required to achieve an $\epsilon$-close approximation of the desired trace we make use of a well-known concentration inequality presented in Lemma \ref{thm-hoeffding}. For more background on stochastic trace estimation and its applications, we refer~\cite{Avron2011, ubaru2017applications,chen2022randomized}. 

\begin{lemma}[Hoeffding's Inequality \cite{Hoeffding1963}]
\label{thm-hoeffding}
Let $X_1,..., X_q$ be independent random variables such that $a_i \leq X_i \leq b_i$ almost surely, and write $s_q = \sum_{i=1}^q X_i$. Then
\[ \Pr \Big[ \big|s_q - \E[s_q]  \big| \geq t \big] \leq 2 \exp\Big(\frac{-2t^2}{\sum_{i=1}^q(b_i -a_i)^2} \Big) \]
\end{lemma}

In this section, we review two modifications of this approach which are used in normalised Betti number estimation.

\subsection{Classical stochastic trace estimation for powers of sparse matrices}\label{sec:classical-ste}
In the classical algorithms described in Section~\ref{sec:apers}, we consider computing the trace of some power $M^d$ of a sparse matrix $M$. A random variable is generated by first choosing a random basis element $\ket{x}$ of the space acted on by $M$ and sampling from a Markov chain computed from the columns of $M$. This method, described in detail by Apers et al.~\cite{Apers2023}, is presented below in Algorithm \ref{alg:markov_chain}. 
In this setting, $M$ is a matrix whose rows and columns are indexed by some set of $n$-bit strings and we make two assumptions about it. Firstly, we assume the existence of an efficient algorithm \textsc{RandomRowIndex}$_M$ which generates a random row from this set. This allows us to generate random basis vectors $\ket{x}$ efficiently. The second assumption is that $M$ is $\mathbf{poly}(n)$-sparse which is equivalent to the existence of an efficient function \textsc{SparseRow}$_M$ which for any row index $i$ returns the row $M_{i,\cdot}$ as a $\mathbf{poly}(n)$-sized dictionary.
This sparsity allows us to generate an unbiased estimate of $\bra{x}M^d\ket{x}$ as follows. Firstly, we create a Markov Chain on the rows of the matrix with transition probabilities $P(x_i, x_j) = |M_{x_ix_j}|/ ||M_{x_i,\cdot}||_1$. We can then define a random variable $$Y_d = \braket{x_d|x_0}\prod_{j=0}^{d-1} \text{sign}(M_{x_{j}x_{j+1}})||M_{x_j,\cdot}||_1 $$  

\noindent
where $x_0$ is a random row generated by \textsc{RandomRowIndex$_M$} and $x_1, x_2, \ldots, x_d$ are successive random rows from the Markov Chain starting at $x_0$. This is an unbiased estimate of $\bra{x_0}M^d\ket{x_0}$ in the sense that $\mathbb{E}[Y_d] = \mathbb{E}[\bra{x_0}M^d \ket{x_0}]$. Furthermore, the norm can be bounded as 

$$|Y_d| = \prod_{j=0}^{d-1}||M_{x_j,\cdot}||_1 \leq \|M\|^d_1.$$

\noindent
For more details on this process and its analysis, see \cite{Apers2023}. As we see in Section \ref{sec:apers}, this bound can be used with Hoeffding's inequality to prove a bound on the required samples to estimate $\frac{1}{N}\tr(M^d)$.

In order to compare algorithms which use Algorithm~\ref{alg:markov_chain} as a subroutine we introduce the notion of classical sample complexity and classical query complexity. For such an algorithm \texttt{A}, the sample complexity will be the number of samples generated throughout the algorithm using Algorithm~\ref{alg:markov_chain} and the query complexity will be the number of calls to the \texttt{SparseRow}$_M$ function which queries a row of the matrix $M$.  As the \texttt{SparseRow}$_M$ oracle will take a fixed amount of time to run depending on the input matrix, we use the query complexity as an estimate of the cost to run $\texttt{A}$ on a given input. The sample complexity on the other hand gives an indication of how many independent processes are needed in running $\texttt{A}$, which can potentially be done in parallel.  

\begin{definition}
    Let \texttt{A} be any algorithm which makes $N$ calls to \texttt{EstimateSparseTrace} with inputs 
    \[ \{(\texttt{SparseRow}_M,d_i, q_i ,\texttt{RandomRowIndex}_M)\}_{1 \leq i\leq N} \] 
    where $N, d_i$ and $q_i$ are, in general, functions of the inputs of \texttt{A}. We define the \emph{classical sample complexity of \texttt{A}} to be the total number $\sum_i q_i$ of all samples generated using \texttt{EstimateSparseTrace}. We define the \emph{classical query complexity of \texttt{A}} to be the total number $\sum_i d_iq_i$ of all queries to \texttt{SparseRow}$_M$ in the algorithm.  
\end{definition}

\begin{algorithm}[h!]
\caption{Markov Chain Trace Estimation of a sparse matrix $M$,  Apers, Gribling, Sen, Szab\'{o}\cite{Apers2023}}\label{alg:markov_chain}
\begin{algorithmic}[1]
\Let {$M$ be a $N\times N$ sparse matrix, where the rows are indexed by a set $S$ of $n$-bit strings.} 
\Require{A function \Call{SparseRow$_M$}{} as described in Section \ref{sec:trace_estimation}.}
\Require {$d$, positive integer denoting the power $M$ is raised to. }
\Require{$q$, the number samples to be taken.}
\Require {A function \Call{RandomRowIndex$_M$}{} which efficiently generates a random element of $S$.}
\Ensure {An estimate for the normalised trace $\frac{1}{N}\tr(M^d)$.}

\Procedure{EstimateSparseTrace}{\textsc{SparseRow$_M$}, d, q, \textsc{RandomRowIndex$_M$}}

\For{$l = 1,\ldots, q$}
    $i_0 \gets$ \textsc{RandomRowIndex$_M$}()
    \For{$k = 0, \ldots, d$}
    
    \State Select a new row index $i_{k+1}$ from the set \textsc{SparseRow$_M$}($i_k$) with probability $P(i_{k+1}, i_k) = M_{i_k, i_{k+1}}/\|M_{i_k, \cdot}\|_1$ 
    \EndFor
    
    \LineComment{Define a value, $s_l$, which approximates $M_{i,i} = \bra{i}M^d\ket{i}$ as follows.}
    \If{$i_{d} = i_0$}
    
    \State $s_l \gets \prod^{j<d}_{j=0} \mathbf{sign}(M_{i_j,i_{j+1}})\|M_{i_j, \cdot}\|_1$
    
    \Else 
    
    \State $s_l \gets 0$
    
    \EndIf
\EndFor

\State \Return $\frac{1}{q}\sum_{l=1}^q s_l$
\EndProcedure
\end{algorithmic}
\end{algorithm}

\subsection{Quantum stochastic trace estimation for positive semi-definite matrices}\label{sec:quantum-ste}
Quantum algorithms for estimating the trace of a unitary matrix have a long history in the field of quantum algorithms, particularly popularised with the Hadamard test used by Aharonov, Jones and Landau in their work on estimating the Jones polynomial~\cite{Aharonov2009}. In recent work on quantum algorithms for Betti numbers, Akhalwaya et al.~\cite{Akhalwaya2022} described an alternative method which is specialised to positive semi-definite matrices which admit a form of block-encoding. This procedure, summarised in Algorithm~\ref{alg:block_encoded_trace}, is central to the quantum algorithms in this paper.

A Hermitian $N \times N$ matrix $M$ is said to be \emph{positive semi-definite} if all of its eigenvalues are real and nonnegative. This is equivalent to the existence of a matrix $D$ such that $M = D^{\dagger}D$. For any such matrix we can rewrite terms of the form $\bra{x}M\ket{x}$ as $\bra{x}D^{\dagger}D\ket{x} = \|D\ket{x}\|^2$, and therefore the trace of $M$ can be expressed as $\tr(M)=\sum_{x} \|D\ket{x}\|^2$. For the trace of higher powers of $M$, say $M^d$, we achieve a similar expression.

\begin{lemma}
\label{lem:psd_mat_decomp}
    If $M$ is a Hermitian positive semi-definite matrix written as $M=D^\dagger D$, then for every positive integer $d$ the trace of $M^d$ is equivalently
\begin{equation*}
    \tr(M^d) = \sum_{x} \| D^{(d)}\ket{x} \|^2,
\end{equation*}
where $D^{(d)}$ is defined as $(D^\dagger D)^{d/2}$ for $d$ even, and $(DD^\dagger)^{\frac{d-1}{2}}D$ for $d$ odd. 
\end{lemma}

\begin{proof}
    If $d$ is even then we have $\bra{x}M^d\ket{x} = \bra{x}(D^\dagger D)^d \ket{x} = \|(D^\dagger D)^{d/2}\ket{x}\|^2$, and for $d$ odd we similarly have
    $\bra{x}M^d\ket{x} = \bra{x}D^\dagger(DD^\dagger)^{d-1}D \ket{x} = \|(DD^\dagger)^{\frac{d-1}{2}}D\ket{x}\|^2$. 
\end{proof}

This fact allows us to create an unbiased trace estimator using a particular type of quantum circuit encoding $D$, as we next explain.

\begin{definition}[Block-encoding]
For an $N\times N$ matrix $D$ whose rows and columns are indexed by a subset $S$ of $n$-bit strings, a \emph{block-encoding} of $D$ with $a$ auxiliary qubits is a $2^{n+a}\times 2^{n+a}$ unitary matrix $U_D$ with the property that, for any $n$-bit string $x \in S$, we have
\begin{equation}
    U_D\ket{0^a}\ket{x} = \ket{0^a} D\ket{x} + \ket{\psi},
\end{equation}
for some state $\ket{\psi}$ that is orthogonal to $\ket{0^a}\ket{y}$ for all $y \in S$. 
\end{definition}
Now observe that for a given bitstring $x \in S$, if we create the state
\begin{equation*}
     \ket{0^a}D^{(d)}\ket{x} + \ket{\psi}
\end{equation*}
for some state $\ket{\psi}$ orthogonal $\ket{0^a}\ket{y}$ for all $y \in S$, then after measuring the auxiliary qubits we obtain the outcome $0^a$ with probability $\|D^{(d)}\ket{x}\|^2$. If we repeat this process for a uniformly random bitstring $x \in S$ then the expected value is $\sum_x \frac{1}{N}\|D^{(d)}\ket{x}\|^2 = \frac{1}{N}\tr(M^d)$, as desired. However, since we only have access to a block-encoding of $D$ or $D^\dagger$, it is not immediately clear how to prepare this state. In Algorithm~\ref{alg:block_encoded_trace} we instead successively apply the block-encoding of $D$ or $D^\dagger$, measure the auxiliary register, and only proceed if the outcome is $0^a$. As we show in Lemma~\ref{lem:quantum_trace_est_alg_correctness}, the probability of succeeding after $d$ of these steps is precisely $\frac{1}{N}\tr(M^d)$. The details of this process are formalised in Algorithm \ref{alg:block_encoded_trace}, and we next show that it indeed estimates the desired trace. For a proof of this, see Appendix~\ref{app:L3.3}.

\begin{lemma}
\label{lem:quantum_trace_est_alg_correctness}
    For each $i=1,...,q$, the random variable $s_i$ output by Algorithm~\ref{alg:block_encoded_trace} is a Bernoulli random variable with expectation $\frac{1}{N}\tr(M^d)$. 
\end{lemma}

In order to compare quantum algorithms which use Algorithm \ref{alg:block_encoded_trace} as a subroutine we introduce the notion of quantum sample complexity and quantum query complexity. For such an algorithm \texttt{B}, the sample complexity will be the number of samples generated throughout the algorithm using Algorithm \ref{alg:block_encoded_trace}, and the query complexity will be the number of calls to the either of the block-encodings $U_D$ or $U_{{D^{\dagger}}}$. As this oracle will have a fixed circuit size depending on the input matrix, we use the query complexity as an estimate of the cost to run $\texttt{B}$ on a given input. The sample complexity on the other hand gives an indication of how many independent quantum processes are needed to run $\texttt{B}$, which can potentially be done in parallel. 

\begin{definition}
    Let \texttt{B} be any algorithm which makes $N$ calls to \texttt{EstimateBlockEncoding} with inputs 
    \[\{(U_D,d_i, q_i ,\texttt{RandomRowIndex}_M)\}_{1 \leq i\leq N}\]
    where $N, d_i$ and $q_i$ are, in general, functions of the inputs of \texttt{A}. We define the \emph{quantum sample complexity of \texttt{A}} to be the total number $\sum_i q_i$ of all samples generated using \texttt{EstimateFromBlockEncoding}. We define the \emph{quantum query complexity of \texttt{A}} to be the total number $\sum_i d_iq_i$ of all queries to $U_D$ or $U_{D^{\dagger}}$ in the algorithm.  
\end{definition}

When $M=D^{\dagger}D$ also admits a classical oracle \texttt{SparseRow}$_M$, these measures are good analogues for the classical sample and query complexities defined in Section \ref{sec:classical-ste}. This is because the query complexity captures the number of uses of a fixed-cost oracle (\texttt{SparseRow}$_M$ in the classical case and $U_D$ in the quantum case) and the sample complexity captures the number of quantum or classical processes that are employed to generate samples.

\begin{algorithm*}[h!]
\caption{Quantum Trace Estimation of a Positive Semi-Definite Matrix with Block Encoding \cite{Akhalwaya2022}}\label{alg:block_encoded_trace}
\begin{algorithmic}[1]
\Let {$M$ be a $N\times N$ positive semi-definite matrix with $M = D^{\dagger} D$ whose rows are indexed by a set $S$ of $n$-bit strings, and \textsc{RandomRowIndex$_M$} a function which efficiently generates a random element of $S$.}
\Require {$U_D$ circuit with $n+a$ qubits which block-encodes $D$; a positive integer $d$ denoting the power $M$ is raised to; a number of samples $q$ }

\Ensure An estimate for the normalised trace $\frac{1}{N}\tr(M^d)$.

\Procedure{EstimateFromBlockEncoding}{$U_D, d, q,$\textsc{RandomRowIndex$_M$}}
\For{$i = 1, \ldots q$}
\State $x \gets $ \textsc{RandomRowIndex$_M$}()
\State $\ket{\phi} \gets \ket{0^a}\ket{x}$ \Comment{Create the initial quantum state from the $n$-bit string $x$.}
\For{$j= 1, \ldots, d$}
\If{$j$ is odd}
\State $\ket{\phi} \gets U_D\ket{\phi}$ \Comment{Apply the quantum circuit $U_D$ to $\ket{\phi}$}
\Else
\State $\ket{\phi} \gets U_{D^{\dagger}}\ket{\phi}$ \Comment{$U_{D^{\dagger}}$ can be constructed as $(U_{D})^{\dagger}$.}
\EndIf
\State Measure auxiliary qubits of $\ket{\phi}$ in computational basis
\If{outcome is $0^a$}
\State \textbf{Continue} to next $j$
\Else
\State $s_i \gets 0$; \textbf{Continue} to next $i$ 
\EndIf
\EndFor
\State $s_i \gets 1$

\EndFor
\State \Return $\frac{1}{q}\sum_{i=1}^q s_i$

\EndProcedure
\end{algorithmic}
\end{algorithm*}


\section{The \ibmOneName{} algorithm}\label{sec:nisq-tda}

\subsection{Outline}\label{sec:nisq-tda-outline}

To our knowledge the first proposed quantum algorithm for normalised Betti number estimation which did not use primitives such as Hamiltonian evolution and phase estimation was presented in the work of Akhalwaya, Ubaru et al.\ across a number of papers~\cite{Ubaru2021,Akhalwaya2022,Akhalwaya2022-2} \footnote{Culminating in a recent ICLR paper \cite{akhalwaya2024topological}, which forward references this paper for the revised complexity analysis of their Monte Carlo  algorithm.}. This work introduced several innovations which opened up the possibility of performing normalised Betti number estimation on near-term devices. Their proposed algorithm works by first choosing a polynomial $p(x)=\sum_{i=0}^d a_i x^i$ such that the trace of $p(\tilde{\Delta}_k)$ is approximately $\beta_k / |S_k|$. They provide a modular circuit for block-encoding $\tilde{\Delta}_k$, which we review in Section \ref{sec:block_encoding}. This block-encoding is then used to perform stochastic trace estimation on the powers $\tilde{\Delta}_k^i$ via Algorithm~\ref{alg:block_encoded_trace}. The traces of the matrices $\tilde{\Delta}_k, \tilde{\Delta}_k^2,...,\tilde{\Delta}_k^d$ are then summed according to the polynomial $p$ to give the normalised Betti number estimate. The full algorithm is summarised in Algorithm \ref{alg:cap}. 

\begin{algorithm}
\caption{The \ibmOneName{} Algorithm  ~\cite{Ubaru2021,Akhalwaya2022}}\label{alg:cap}
\begin{algorithmic}[1]
\Require {$\mathcal{G}$, a graph on $n$ vertices with clique complex $\Gamma$; a desired dimension $k$ such that $1 \leq k \leq n-1$; a desired error $\epsilon >0$; a desired failure probability $\eta >0$}
\Ensure {An $(\epsilon, \eta)$-estimator $\hat{\beta}_k$ of the $k^{\textit{th}}$ normalised Betti number of $\Gamma$.}
\Let $\normalLap$ be the normalised $k^{\textit{th}}$ combinatorial Laplacian of $\Gamma$.
\Let  $R_{k}$ be an efficient random sampler of $k$-simplices from $\Gamma$. 
\Let $\delta \in (0, 1/2]$ a lower bound on smallest positive eigenvalue of  $\normalLap$.
\Procedure{\ibmOneName}{}
\State Define the $(\epsilon/2, \delta)$-filter polynomial $p(x)=\sum_{i=0}^d a_i x^i$ as described in Lemma~\ref{lem-chebypoly}.
\State $d \gets \frac{1}{\sqrt{\delta}} \log(4/\epsilon)$
\State $q \gets \lceil \frac{2\log(2/\eta)}{\epsilon^2}\|p\|_2^2 \rceil $ \Comment{See Theorem \ref{thm-IBM_betti_estimate}}
\State Let $U_D$ be the block encoding of a matrix $D$ such that $D^\dagger D = \normalLap$
\For{i = 1, \ldots, d}
\State $\hat{\mu}^{(i)} \gets $\Call{EstimateFromBlockEncoding}{$U_D$, i, q, $R_k$}
\EndFor
\State \Return $a_0 + a_1\hat{\mu}^{(1)} + \ldots + a_d \hat{\mu}^{(d)}$ 
\EndProcedure
\end{algorithmic}
\end{algorithm}

In this section we describe the methods used in each of these steps and give a new assessment of the time complexity of this algorithm~\cite{Akhalwaya2022}. In particular, the complexity analysis we give in Section~\ref{sec:IBM_complexity} shows that the number of uses of a block-encoding of $\tilde{\Delta}_k$ scales with the $2$-norm of the polynomial $p(x)$ chosen above. We additionally show in Section~\ref{sec:polyapprx} that the polynomial considered in \cite{Akhalwaya2022} has $2$-norm that grows exponentially in its degree, which then leads to a term in the algorithm's complexity scaling exponentially with $1/\delta$.

\subsection{Quantum circuits for the Laplacian}\label{sec:block_encoding}

The quantum circuit used in this algorithm is built from a number of simple operators considered by Akhalwaya et al. that we summarise now \cite{Akhalwaya2022, Akhalwaya2022-2}. First, the normalised Laplacian can be expressed as a product of operators as
\begin{equation}\label{eqn:decomp}
    \tilde{\Delta}_k = P_{k}P_{\Gamma}\Big(\frac{1}{\sqrt{n}}B\Big) P_{\Gamma} \Big(\frac{1}{\sqrt{n}}B\Big) P_{\Gamma} P_{k}.    
\end{equation}
Here, $B=\partial + \partial^{\dagger}$ denotes the unrestricted boundary operator, $P_{\Gamma} = \sum_{\sigma\in \Gamma} \ketbra{\sigma}{\sigma}$ denotes the projection onto the simplices of the clique complex, and $P_k$ is the projection onto the Hamming weight $k$ subspace of the $n$-qubit Hilbert space. The operator $\frac{1}{\sqrt{n}}B$ is both Hermitian and unitary, and a quantum circuit construction for it using the Jordan-Wigner transform was given in \cite{Akhalwaya2022-2}. Additionally, block-encodings for the projections $P_\Gamma$ and $P_k$ are described in \cite{Akhalwaya2022} using a circuit of Toffoli gates and a Quantum Fourier transform, respectively. The identity in (\ref{eqn:decomp}) allows us to express the normalised Laplacian in the form $\tilde{\Delta}_k =\tilde{B}^\dagger \tilde{B}$, where $\tilde{B}$ is called the restricted boundary operator. This form allows for the use of quantum stochastic trace estimation algorithm described in Algorithm~\ref{alg:block_encoded_trace} by setting $D=\tilde{B}=P_{\Gamma} \big(\frac{1}{\sqrt{n}}B\big) P_{\Gamma} P_{k}$. The circuit implementation of $\frac{1}{\sqrt{n}}B$ and block-encodings of $P_\Gamma$ and $P_k$ previously mentioned yield a block-encoding $U_D$ of $D$ by successively applying each of the circuits for $P_k, P_\Gamma$ and $\frac{1}{\sqrt{n}}B$, measuring the ancilla register in between each application and proceeding only if the measurement outcome is $0^a$. For higher powers of the normalised Laplacian, say $\tilde{\Delta}_k^d$, we can write $\tilde{\Delta}_k^d = (D^{(d)})^{\dagger}D^{(d)}$ using the construction discussed in Section~\ref{sec:trace_estimation} to obtain a block-encoding $U_{D^{(d)}}$ of $D^{(d)}$.

\subsection{Polynomial constructions}
\label{sec:polyapprx}

In this section we describe a family of polynomials considered in the Betti number estimation algorithm of \cite{Akhalwaya2022} that is used in Algorithm~\ref{alg:cap}. We first give general conditions for a polynomial $p(x)$ to yield an appropriate approximation of the normalised Betti number when applied in the context of Algorithm~\ref{alg:cap}. We will then show that the polynomials used in \cite{Akhalwaya2022} have $2$-norm growing exponentially with their degree, which we then show leads to an exponential term in the algorithm's time-complexity. 

\begin{lemma}
\label{lem:filter_poly}
    Let $\epsilon > 0$ be given, and let $\delta$ denote the smallest positive eigenvalue of the normalized Laplacian $\tilde{\Delta}_k$. If $p(x)$ is a real polynomial with the property that $p(0)=1$ and $|p(x)|<\epsilon$ for all $x \in [\delta, 1]$, then 
\begin{equation}
    \Big|\frac{1}{|S_k|}\tr\big(p(\tilde{\Delta}_k)\big) - \frac{1}{|S_k|}\beta_k\Big| < \epsilon.
\end{equation}
\end{lemma}

\begin{proof}
Since the trace of a diagonalisable matrix is the sum of its eigenvalues, and since $\beta_k$ is the dimension of the kernel of $\tilde{\Delta}_k$, it follows that
\begin{equation*}
    \tr\big(p(\tilde{\Delta}_k)\big) = \beta_k + \sum_{\lambda > 0} p(\lambda),
\end{equation*}
where the above sum is over the positive eigenvalues of $\tilde{\Delta}_k$, counting multiplicity. As $|p(x)| < \epsilon$ for $x \in [\delta, 1]$ and $\tilde{\Delta}_k$ is an $|S_k| \times |S_k|$ matrix then we obtain the upper and lower bounds
\begin{equation*}
    \beta_k - \epsilon |S_k| \leq \tr\big(p(\tilde{\Delta}_k)\big) \leq \beta_k + \epsilon |S_k|,
\end{equation*}
or equivalently
\begin{equation}
    \Big|\frac{1}{|S_k|}\tr\big(p(\tilde{\Delta}_k)\big) - \frac{1}{|S_k|}\beta_k\Big| < \epsilon.
\end{equation}
\end{proof}

Lemma~\ref{lem:filter_poly} shows that we can estimate the normalised Betti number by an estimation of the normalised trace of $p(\tilde{\Delta}_k)$ for an appropriately chosen polynomial $p(x)$. By linearity of the trace, the normalised trace of $p(\tilde{\Delta}_k)$ can be estimated by estimating the normalised traces of $\tilde{\Delta}_k,\tilde{\Delta}_k^2,...,\tilde{\Delta}_k^d$ and summing according to the coefficients of $p(x)$. The conditions of Lemma~\ref{lem:filter_poly} motivates the following definition.

\begin{definition}
    We will say that a real function $f(x)$ is an $(\epsilon,\delta)$-\textit{filter} if $f(0)=1$ and $|f(x)|<\epsilon$ for all $x \in [\delta, 1]$.
\end{definition}

In \cite{Akhalwaya2022}, Akhalwaya et al. consider a polynomial filter based on the Chebyshev polynomials, which we now describe. We let $T_d(x)$ denote the $d^{th}$ degree Chebyshev polynomial of the first kind, defined by the recurrence $T_0(x)=1$, $T_1(x)=x$ and
\begin{equation}
\label{eqn:chebyrecurrence}
  T_{d+1}(x)=2xT_d(x)-T_{d-1}(x),  
\end{equation}
for $d \geq 1$. It was shown in \cite[Proposition 1]{Akhalwaya2022} that by shifting and scaling a Chebyshev polynomial of sufficiently large degree we obtain an $(\epsilon, \delta)$-filter. 

\begin{lemma}
\label{lem-chebypoly}
Let $\epsilon, \delta > 0$ be given. For all $d \geq \frac{1}{\sqrt{\delta}} \log(2/\epsilon)$ the polynomial
\begin{equation}
\label{eqn:chebypoly}
    T_d\Big(\frac{1-x}{1-\delta}\Big)/T_d\Big(\frac{1}{1-\delta}\Big),
\end{equation}
is an $(\epsilon, \delta)$-filter.
\end{lemma}

Define the $2$-norm of a polynomial $p(x)=a_0 + a_1 x + ... + a_d x^d$ as $\|p\|_2 := \sqrt{a_0^2 + a_1^2+...+a_d^2}$. Next we show that the $2$-norm of the polynomial in Lemma~\ref{lem-chebypoly} is exponentially large in its degree. We defer the proof to the Appendix~\ref{app:L4.3}.

\begin{lemma}
\label{lem-coef_bound}
Let $\epsilon > 0$ and $\delta \in (0, 1/2]$ be given, and let $p(x)$ denote the polynomial 
\begin{equation}
    p(x)=T_d\Big(\frac{1-x}{1-\delta}\Big)/T_d\Big(\frac{1}{1-\delta}\Big).
\end{equation}
For $d \geq \frac{1}{\sqrt{\delta}} \log(2/\epsilon)$, the polynomial $p(x)$ has $2$-norm which satisfies
\begin{equation*}
  \frac{1}{d+1}(9/4)^d \leq \|p\|_2^2 \leq \epsilon^2 d 2^{10d}.
\end{equation*}
\end{lemma}
\subsection{Complexity analysis}
\label{sec:IBM_complexity}

In this section we prove the correctness and give the complexity of the \ibmOneName{} algorithm described in Algorithm \ref{alg:cap}.
Consider the quantities $\hat{\mu}^{(j)}$ as defined in Algorithm~\ref{alg:cap}. Each $\hat{\mu}^{(j)}$ can be expressed as an average of random variables, say $\hat{\mu}^{(j)} = \frac{1}{q}\sum_{i=1}^q \hat{\mu}_i^{(j)}$, where $\hat{\mu}_i^{(j)}$ are random variables taking value $0$ or $1$, as described in Algorithm~\ref{alg:block_encoded_trace}. 
The normalised Betti number estimate resulting from Algorithm~\ref{alg:cap} is
\begin{equation}
\label{eqn:betti_estimate1}
    \hat{\beta}_k = a_0 + a_1\hat{\mu}^{(1)} + \ldots + a_d \hat{\mu}^{(d)}.
\end{equation}
We can now give the number of samples $q$ required for this estimator to be an $(\epsilon, \eta)$-estimator.

\begin{theorem}
\label{thm-IBM_betti_estimate}
    Let $\epsilon, \eta > 0$ be given, and suppose $p(x) = a_0 + a_1 x + ... + a_d x^d$ is a real polynomial that is an $(\epsilon/2, \delta)$-filter. The normalised Betti number estimate $\hat{\beta}_k$ output of Algorithm~\ref{alg:cap} is an $(\epsilon, \eta)$-estimator.
\end{theorem}

\begin{proof}
    Applying the triangle inequality and Lemma~\ref{lem:filter_poly} we obtain
\begin{equation*}
\begin{split}
    \Big|\hat{\beta}_k - \frac{\beta_k}{|S_k|}\Big| &\leq  \Big|\hat{\beta}_k - \frac{1}{|S_k|}\tr\big(p(\tilde{\Delta}_k)\big) \Big| + \Big| \frac{1}{|S_k|}\tr\big(p(\tilde{\Delta}_k)\big) - \frac{\beta_k}{|S_k|}\Big| \\
    &\leq \Big|\hat{\beta}_k - \frac{1}{|S_k|}\tr\big(p(\tilde{\Delta}_k)\big) \Big| + \frac{\epsilon}{2}, \\
\end{split}
\end{equation*}
and therefore we have the lower bound
\begin{equation}
\label{eqn:prob_bound}
    \Pr\Big[ \Big|\hat{\beta}_k - \frac{\beta_k}{|S_k|}\Big| \leq \epsilon \Big] \geq \Pr\Big[ \Big|a_0 + \frac{1}{q}\sum_{i=1}^q  \sum_{j=1}^d a_j \hat{\mu}_i^{(j)} - \frac{1}{|S_k|}\tr\big(p(\tilde{\Delta}_k)\big) \Big| \leq \frac{\epsilon}{2}\Big].
\end{equation}
The random variables $a_j\hat{\mu}_{i}^{(j)}$ are independent and have absolute value in the interval $[0, |a_j|]$, hence applying Lemma~\ref{thm-hoeffding} with $t=\epsilon q /2$ yields
\begin{equation*}
    \Pr\Big[ \Big|a_0 + \frac{1}{q}\sum_{i=1}^q  \sum_{j=1}^d a_j \hat{\mu}_i^{(j)} - \frac{1}{|S_k|}\tr\big(p(\tilde{\Delta}_k)\big) \Big| \geq \epsilon/2 \Big] \leq 2 \exp\Big(\frac{-\epsilon^2 q}{2\|p\|_2^2} \Big).
\end{equation*}
For $q \geq \frac{2\log(2/\eta)}{\epsilon^2}\|p\|_2^2$ the right side of the above is at most $\eta$, and combined with (\ref{eqn:prob_bound}) gives the result.
\end{proof}

Lastly, we give the sample and query complexity of Algorithm~\ref{alg:cap}.

\begin{theorem}\label{thm:complexity_ibm_one}
Let $G$ be a given graph on $n$ vertices, $k$ a desired dimension, $\epsilon$ a desired error, $\eta$ a desired failure probability, and let $\delta$ be the spectral gap of the normalized Laplacian $\tilde{\Delta}_k$. Executing Algorithm \ref{alg:cap} gives us an $(\epsilon, \eta)$-estimator of $\beta_k/|S_k|$ of the clique complex of $G$, with sample complexity 
\begin{equation*}
    S_{\text{\ibmOneName}} =  \log(2/\eta) \times \frac{\log(4/\epsilon)}{\sqrt{\delta}}   \times 2^{\frac{10\log(4/\epsilon)}{\sqrt{\delta}}}  
\end{equation*}
and query complexity 
\begin{equation*}
    Q_{\text{\ibmOneName}} = \log(2/\eta) \times \Big(\frac{1}{\sqrt{\delta}} \log(4/\epsilon)\Big)^2 (\frac{1}{\sqrt{\delta}} \log(4/\epsilon) + 1 \Big) \times 2^{\frac{10 \log(4/\epsilon)}{\sqrt{\delta}} }
\end{equation*}
\end{theorem}

\begin{proof}
Following Algorithm~\ref{alg:cap} there are $d$ estimates $\hat{\mu}^{(i)}$ produced, each of which uses $q$ samples. Therefore the total sample complexity is
\begin{equation*}
   d \times q = \frac{1}{\sqrt{\delta}} \log(4/\epsilon) \times \frac{2 \log(2/\eta)}{\epsilon^2} \times \|p \|_2^2 \leq \frac{1}{\delta} \log(4/\epsilon)^2 \times 2 \log(2/\eta) \times  2^{10d},
\end{equation*}
where we have applied Lemma~\ref{lem-coef_bound} to bound the $2$-norm of the polynomial used in this case. For the query complexity, we count the number of uses of the circuit $U_D$ or its Hermitian conjugate. In each sample applied towards the estimate $\hat{\mu}^{(i)}$ we use $i$ uses of the circuit $U_D$ or its Hermitian conjugate. Therefore the total query complexity is
\begin{equation*}
    \sum_{i=1}^d q \times i = q \times \frac{d(d+1)}{2} \leq \frac{2 \log(2/\eta)}{\epsilon^2} \epsilon^2 d 2^{10d} \times \frac{1}{2} \Big(\frac{1}{\sqrt{\delta}} \log(4/\epsilon)\Big) (\frac{1}{\sqrt{\delta}} \log(4/\epsilon) + 1 \Big)
\end{equation*}
\end{proof}

We note that the term $2^{\frac{10\log(4/\epsilon)}{\sqrt{\delta}}}$ in the complexity of Theorem~\ref{thm:complexity_ibm_one} is not present in the analyses given in previous work.

\section{The classical BNE algorithms}\label{sec:apers}

In this section, we recall the two classical algorithms of Apers, Gribling, Sen and Szabó~\cite{Apers2023} for estimating Betti numbers. In this original work, the authors consider a broader problem than that addressed in this paper. In particular, their algorithm is described for all finite simplicial complexes (not just Vietoris-Rips complexes) and they are able to exploit an upper bound $\hat{\lambda}$ on the eigenvalues of $\Delta_k$. In order to compare these algorithms directly with their quantum counterparts, Algorithm \ref{alg:cap} and \ref{alg:this_paper}, we limit the scope of these algorithms in this section to that of Algorithm \ref{alg:cap}. We also present some small improvements to the design of this algorithms which help to present a fairer comparison in Section \ref{sec:num_exp}.     

\subsection{Outline}

The algorithms presented by Apers et al.~\cite{Apers2023} have a very similar structure to Algorithm \ref{alg:cap}
presented in the last section. In particular, the $k^{\textit{th}}$ normalised Betti number is approximated
by the normalised trace $\frac{1}{|S_k|}\tr(p(M))$ for some polynomial $p$ and relevant matrix $M$, then this 
quantity is estimated by stochastic trace estimation on the relevant powers of $M$. There are two main differences in our presentation of this
algorithm versus the original paper. Firstly, the matrix $M$ taken by Apers et al.\ is the \emph{reflected Laplacian} $H = I - \tilde{\Delta}_k$ instead of 
$\tilde{\Delta}_k$\footnote{In the original presentation, the matrix $H$ is defined with $\tilde{\Delta}_k = \Delta_k/\hat{\lambda}$ where 
$\hat{\lambda}\leq n$ is an upper bound on the largest eigenvalue of $\Delta_k$.}. The eigenvalues of $\tilde{\Delta}_k$ all fall in the range $[0,1]$, 
as noted in Section \ref{sec:top_background}. Thus the eigenvalues of $I-\tilde{\Delta}_k$ are confined to the same range and 
the Hodge theorem \cite{Eckmann1944} implies that the dimension of the $1$-eigenspace of $H$ is equal to the 
$k^{\textit{th}}$ Betti number. As we show, this changes the polynomials $p$ which are needed for these algorithms.
Secondly, the algorithm employs the classical Monte Carlo method of stochastic trace estimation described in Algorithm \ref{alg:markov_chain}. To use this they show that $H$ is sparse via the Laplacian Matrix Theorem~\cite[Theorem 3.3.4]{Goldberg2002}. This theorem can be used to implement the function \textsc{SparseRow}$_M$ in the course of the numerical simulations presented in Section \ref{sec:num_exp}. 

\subsection{\apersOneName{}  algorithm}

The first classical Betti number estimation algorithm of \cite{Apers2023} which we present in modified form as Algorithm \ref{alg:apers_1}, estimates the normalised Betti number in two simple steps. Firstly, it is observed that the desired quantity can be estimated to any accuracy $\epsilon$ by $\tr(H^d)/|S_k|$ for sufficiently high $d$. This is an observation made originally by Friedman~\cite{Friedman1998} and we reprove the exact relationship between $\epsilon$ and $d$ for completeness. 

\begin{lemma}
\label{lem-power_method_exp}
    For any $\epsilon > 0$, if $d \geq \frac{\log(1/\epsilon)}{\delta}$ then the normalised trace of $H^d$ satisfies
\begin{equation}
    \frac{1}{|S_k|}\beta_k \leq \frac{1}{|S_k|} \tr(H^d) \leq \frac{1}{|S_k|} \beta_k + \epsilon.
\end{equation}
\end{lemma}

\begin{proof}
The trace of $H^d$ is evaluated as
\begin{equation}
\tr(H^d) = \beta_k + \sum_{\lambda > 0} (1 - \lambda/n)^d
\end{equation}
where the sum is over the positive eigenvalues $\lambda$ of $\Delta_k$, including multiplicity. Since the eigenvalues of $\Delta_k$ lie in the interval $[0,n]$, then each term $(1 - \lambda/n)^d$ above is nonnegative, and hence $\beta_k \leq \tr(H^d)$. On the other hand, since $\delta$ is a lower bound for the positive eigenvalues of $\tilde{\Delta}_k$ then each term $(1 - \lambda/n)^d$ above is at most $(1-\delta)^d$. Choosing $d \geq \frac{\log(1/\epsilon)}{\delta}$ ensures that $(1-\delta)^d \leq \epsilon$, thus we obtain $\tr(H^d) \leq \beta_k + \epsilon |S_k|$, completing the proof.
\end{proof}

Given this approximation, the algorithm then estimates $\tr(H^d)/|S_k|$ using the classical Markov chain trace estimation described in Algorithm~\ref{alg:markov_chain}. As was shown in \cite{Apers2023}, the number of samples required to obtain an $\epsilon$-estimate of $\tr(H^d)/|S_k|$ with probability $1-\eta$ can then be deduced from Lemma \ref{thm-hoeffding}. We give a short proof for completeness. 

\begin{lemma}
\label{lem:apers1_correctness}
    Let $\epsilon, \eta > 0$ be given. The output of Algorithm~\ref{alg:apers_1} is an $(\epsilon, \eta)$-estimator of $\beta_k/|S_k|$.
\end{lemma}

\begin{proof}
    Let $\hat{\beta}_k = \frac{1}{q} \sum_{\ell=1}^q s_\ell$ denote the output of Algorithm~\ref{alg:apers_1}, where $q = \lceil \frac{2^{2d+1} \log(2/\eta)}{\epsilon^2} \rceil$ and the $s_\ell$ are i.i.d random variables described in Algorithm~\ref{alg:markov_chain}. Similar to the proof of Theorem~\ref{thm-IBM_betti_estimate}, an application of the triangle inequality shows
\begin{equation*}
\begin{split}
    \Big|\hat{\beta}_k - \frac{\beta_k}{|S_k|}\Big| &\leq \Big|\hat{\beta}_k - \frac{1}{|S_k|}\tr\big(H^d\big) \Big| + \frac{\epsilon}{2}, \\
\end{split}
\end{equation*}
and therefore we have the lower bound
\begin{equation*}
    \Pr\Big[ \Big|\hat{\beta}_k - \frac{\beta_k}{|S_k|}\Big| \leq \epsilon \Big] \geq \Pr\Big[ \Big|\hat{\beta}_k - \frac{1}{|S_k|}\tr\big(H^d\big)\Big| \leq \frac{\epsilon}{2}\Big].
\end{equation*}
The random variables $s_\ell$ are bounded as $0 \leq s_\ell \leq \|H\|_1^d$ as described in Section~\ref{sec:trace_estimation}. Therefore applying Lemma~\ref{thm-hoeffding} with $t = \epsilon q/2$ gives
\begin{equation}
    \Pr\Big[ \Big|\hat{\beta}_k - \frac{1}{|S_k|}\tr\big(H^d\big)\Big| \geq \frac{\epsilon}{2}\Big] \leq 2 \exp\Big( \frac{-\epsilon^2 q}{2\|H\|_1^{2d}} \Big)
\end{equation}
To make the right side at most $\eta$ it suffices to take $q \geq 2\ln(2/\eta)\|H\|_1^{2d}/\epsilon^2$. As was pointed out in \cite[Corollary 4.4]{Apers2023}, for a clique complex the $1$-norm of the matrix $H$ is at most $2$, therefore choosing $q=\lceil \frac{2^{2d+1} \log(2/\eta)}{\epsilon^2} \rceil$ we obtain
\begin{equation*}
    \Pr\Big[ \Big|\hat{\beta}_k - \frac{\beta_k}{|S_k|}\Big| \leq \epsilon \Big] \geq 1 - \eta
\end{equation*}
as was desired.
\end{proof}

Having shown that the output of Algorithm~\ref{alg:apers_1} is an $(\epsilon, \eta)$-estimator, we can now give the sample and query complexity.

\begin{theorem}\label{thm:complexity_apers1}

Let $G$ be a given graph on $n$ vertices, $k$ a desired dimension, $\epsilon$ a desired error, $\eta$ a desired failure probability, and let $\delta$ be the spectral gap of the normalised Laplacian $\tilde{\Delta}_k$. Executing Algorithm \ref{alg:apers_1} gives an $(\epsilon, \eta)$-estimator of $\beta_k/|S_k|$ of the clique complex of $G$, and with sample complexity
\begin{equation*}
    S_{\text{\apersOneName}} = \frac{1}{\epsilon^2} \times \log(2/\eta) \times 2^{\frac{2\log(2/\epsilon)}{\delta} + 1}
\end{equation*}
and query complexity 
\begin{equation*}
    Q_{\text{\apersOneName}} = \frac{1}{\epsilon^2} \times  \frac{\log(2/\epsilon)}{\delta} \times  \log(2/\eta) \times 2^{\frac{2\log(2/\epsilon)}{\delta} + 1}
\end{equation*}
\end{theorem}

\begin{proof}
With $d = \lceil \frac{\log(2/\epsilon)}{\delta} \rceil$, Algorithm~\ref{alg:apers_1} consists of running Algorithm \ref{alg:markov_chain} to estimate the normalised trace of $H^d$ to accuracy $\epsilon/2$ and with probability $1-\eta$. In Lemma~\ref{lem:apers1_correctness} we showed that this can be done with $\lceil \frac{2^{2d+1} \log(2/\eta)}{\epsilon^2} \rceil$ samples from the random variable in Algorithm \ref{alg:markov_chain}. This gives the claimed sample complexity. For the query complexity, each of these samples required $d$ calls to the \textsc{SparseRow}$_{I-\tilde{\Delta}_k}$ to simulate $d$ steps of the relevant Markov chain. This means that the query complexity is $d$ times the sample complexity, which is precisely the claimed query complexity. 
\end{proof}

\begin{algorithm}[t]
\caption{\apersOneName \cite[Algorithm 1]{Apers2023}}\label{alg:apers_1}
\begin{algorithmic}[1]
\Require {The input is per that of Algorithm \ref{alg:cap}}
\Ensure {The output is per that of Algorithm \ref{alg:cap}}
\Let $\normalLap, R_k$ and $\delta$ be defined as per Algorithm \ref{alg:cap}. Let \textsc{S}$_{H}$ be the function \textsc{SparseRow}$_{I-\tilde{\Delta}_k}$ as described in Section \ref{sec:trace_estimation}.
\Procedure{\apersOneName}{$\mathcal{G}, k, \epsilon$}
\State $d \gets \lceil \frac{\log(2/\epsilon)}{\delta} \rceil $ \Comment{See Lemma \ref{lem-power_method_exp}.} 
\State $q \gets \lceil \frac{2^{2d+1} \log(2/\eta)}{\epsilon^2} \rceil$ \Comment{See Lemma~\ref{lem:apers1_correctness}.}
\State \Return \Call{EstimateSparseTrace}{$S_H$, q, d, $R_{k}$}
\EndProcedure
\end{algorithmic}
\end{algorithm}

\subsection{\apersTwoName\ algorithm}

As we saw in Theorem~\ref{thm:complexity_apers1}, Algorithm \ref{alg:apers_1} has an exponential asymptotic dependence on the degree $d =\frac{1}{\delta} \log(2/\epsilon)$. In \cite{Apers2023}, a second algorithm was proposed which reduces this exponential term to an exponential of $\mathcal{O}\big(\frac{1}{\sqrt{\delta}}\log(1/\epsilon)\big)$ by choosing a polynomial approximation with a lower degree. This second algorithm uses a well-known approximation of the monomial $x^r$ by a sum of Chebyshev polynomials of degrees $1,2, \ldots d$. (For more details, see the exposition of \cite[Theorem 3.2]{sachdeva2013approximation}.) This polynomial is written as $p_{r,d}(x)$ and the important consequence of the theorem cited above is that choosing $d \in \tilde{\mathcal{O}}(\sqrt{r})$ is sufficient to guarantee any constant uniform approximation of $x^r$ in the range $[-1,1]$. This approximation is then used in the same way as the polynomial approximation in Algorithm \ref{alg:cap} in that each trace $\tr(H)/|S_k|, \tr(H^2)/|S_k|, \ldots, \tr(H^d)/|S_k|$ is approximated using stochastic trace estimation as per Algorithm \ref{alg:markov_chain}, and then summed to get an estimate of $\beta_k/|S_k|$. 

We make two observations which improve the analysis of this algorithm. The first is that in general a lower degree polynomial can be used in the approximation compared with $p_{r,d}$. 

\begin{lemma}
\label{lem:cheb_poly_apers2}
    Let $\epsilon > 0$, and let $\delta$ be the smallest nonzero eigenvalue of $\tilde{\Delta}_k$. For $d \geq \log(4/\epsilon)/\sqrt{\delta}$, the polynomial 
\begin{equation*}
    p(x) = T_{d}\Big(\frac{x}{1-\delta}\Big)/T_d\Big(\frac{1}{1-\delta}\Big)
\end{equation*}
satisfies 
\begin{equation*}
    \Big|\frac{1}{|S_k|}\tr\big(p(H)\big)-\frac{\beta_k}{|S_k|} \Big| < \epsilon/2.
\end{equation*}
\end{lemma}

\begin{proof}
    This follows immediately from Lemma~\ref{lem:filter_poly} and \ref{lem-chebypoly} by noticing that the polynomial considered here is a reflection of the polynomial considered there.
\end{proof}

\begin{remark}
    In the original version of \apersTwoName{}\cite{Apers2023} the polynomial $p_{r, d}$ with $r = \lceil \log(3/\epsilon)/\delta \rceil$ and $d =  \lceil\sqrt{2/\delta}\log(6/\epsilon)\rceil$ has degree $d$ and achieves the approximation
\begin{equation*}
    \Big|\frac{1}{|S_k|} \tr\big(p_{r,d}(H)\big)-\frac{\beta_k}{|S_k|} \Big| < 2\epsilon/3.
\end{equation*}
The polynomial $p(x)$ described in Lemma~\ref{lem:cheb_poly_apers2} has lower degree $d = \lceil\log(4/\epsilon)/\sqrt{\delta}\rceil$ and achieves the better approximation
\begin{equation*}
    \Big|\frac{1}{|S_k|}\tr(p(H)) - \frac{\beta_k}{|S_k|} \Big| < \epsilon/2.
\end{equation*}
\end{remark}

\begin{algorithm}[t]
\caption{\apersTwoName \cite[Algorithm 2]{Apers2023}}\label{alg:apers_2}
\begin{algorithmic}[1]
\Require {The input is per that of Algorithm \ref{alg:cap}}
\Ensure {The output is per that of Algorithm \ref{alg:cap}}
\Let $\normalLap, R_k$ and $\delta$ be defined as per Algorithm \ref{alg:cap}, and let \textsc{S}$_{H}$ be the function \textsc{SparseRow}$_{I-\tilde{\Delta}_k}$ as described in Section \ref{sec:trace_estimation}.
\Procedure{\apersTwoName}{$\mathcal{G}, k, \epsilon$}

\State Define polynomial $p(x)=\sum_{i=0}^d a_i x^i$ s.t. $ \big|\tr(p(I - \tilde{\Delta}_k))/|S_k|- \beta_k/|S_k|\big| < \epsilon/2$
\Comment{See Lemma~\ref{lem:cheb_poly_apers2}}
\State $\eta' \gets 1- \sqrt[\lceil d/2 \rceil]{1 - \eta} $ 
\For{$i = 1, \ldots, d$ \textnormal{where} $a_i \neq 0$}
\State $\epsilon_i \gets \frac{\epsilon}{2 \lceil d/2 \rceil |a_i|}$ 
\State $q_i \gets \lceil \frac{2^{2d+1} \log(2/\eta')}{\epsilon_i^2} \rceil$ 
\State $\hat{\mu}^{(i)} \gets $\Call{EstimateSparseTrace}{\textsc{S}$_{H}$, $q_i$, $i$, $R_{k}$} 
\EndFor
\State \Return $a_0 + a_1\hat{\mu}^{(1)} + \ldots + a_d \hat{\mu}^{(d)}$
\EndProcedure
\end{algorithmic}
\end{algorithm}

Similar to Lemma~\ref{lem-coef_bound} we give an upper bound on the 2-norm of the polynomial defined in Lemma~\ref{lem:cheb_poly_apers2} which will later be used to bound the complexity of Algorithm~\ref{alg:apers_2}. The proof can be found in the Appendix~\ref{app:L5.5}.

\begin{lemma}
\label{lem-coef_bound2}
Let $\epsilon > 0$ and $\delta \in (0, 1/2]$ be given, and let $p_d(x)$ denote the polynomial 
\begin{equation*}
    p_d(x) = T_{d}\Big(\frac{x}{1-\delta}\Big)/T_d\Big(\frac{1}{1-\delta}\Big)
\end{equation*}
For $d \geq \frac{1}{\sqrt{\delta}} \log(4/\epsilon)$, the polynomial $p_d(x)$ has $2$-norm which satisfies
\begin{equation*}
  \|p_d\|_2 \leq \epsilon2^{3d-1}.
\end{equation*}
\end{lemma}

In Apers et al.'s original algorithm, the samples for each trace estimation of $\tr(H^i)/|S_k|$ are divided up differently from how this process is done in Algorithm \ref{alg:cap}. In their version, they perform a separate trace estimation for each monomial in $p_{r,d}$ with a separate error $\epsilon_i$ for each such that $\sum \epsilon_i < \epsilon/2$.  Our second observation gives more precise values of these errors and chooses shot counts to ensure that the final estimate is $\epsilon$-close to the $k^{\textit{th}}$ normalised Betti number with confidence $\eta$. We also note that we only need to perform trace estimations for the non-zero monomials of the polynomial $p$. The polynomial  described in Lemma~\ref{lem:cheb_poly_apers2} is either even or odd due to the well-known fact that the Chebyshev polynomials are either even or odd, depending on the parity of its degree. Therefore when applying Algorithm~\ref{alg:apers_2} there are only $\Bar{d} = \lceil d/2 \rceil$ of these traces $\tr(H^i)/|S_k|$ needed to estimate.

\begin{lemma}
\label{lem:apers2_correctness}
    Let $\epsilon, \eta > 0$ be given. The output of Algorithm~\ref{alg:apers_2} is an $(\epsilon, \eta)$-estimator of $\beta_k/|S_k|$.
\end{lemma}

\begin{proof}
    Let $\hat{\beta}_k = a_0 + a_1\hat{\mu}^{(1)} + \ldots + a_d \hat{\mu}^{(d)}$ denote the output of Algorithm~\ref{alg:apers_2}. Similar to the proof of Theorem~\ref{thm-IBM_betti_estimate}, an application of the triangle inequality shows
\begin{equation*}
\begin{split}
    \Big|\hat{\beta}_k - \frac{\beta_k}{|S_k|}\Big| &\leq \Big|\hat{\beta}_k - \frac{1}{|S_k|}\tr\big(p(H)\big) \Big| + \frac{\epsilon}{2}, \\
\end{split}
\end{equation*}
and therefore it remains to show
\begin{equation*}
\begin{split}
    \Pr\Big[ \Big|\sum_{i=1}^d |a_i|\Big( \hat{\mu}^{(i)} - \frac{1}{|S_k|}\tr\big(H^i\big)\Big)\Big| \leq \frac{\epsilon}{2}\Big] \geq 1 - \eta.
\end{split}
\end{equation*}
This probability can be bounded below by the probability that every estimate $\hat{\mu}^{(i)}$ is simultaneously close to its expected value $\frac{1}{|S_k|}\tr\big(H^i\big)$ with high enough probability. More precisely, we have the lower bound
\begin{equation*}
\begin{split}
    \Pr\Big[ \Big|\sum_{i=1}^d |a_i|\Big( \hat{\mu}^{(i)} - \frac{1}{|S_k|}\tr\big(H^i\big)\Big)\Big| \leq \frac{\epsilon}{2}\Big] \geq \prod_{i=0}^d \Pr\Big[ \Big| \hat{\mu}^{(i)} - \frac{1}{|S_k|}\tr\big(H^i\big)\Big| \leq \frac{\epsilon}{2\lceil d/2 \rceil |a_i|}\Big].
\end{split}
\end{equation*}
Each estimate $\hat{\mu}^{(i)}$ is an average of $q_i$ samples which lie between $0$ and $\|H\|_1^{2d} \leq 2^{2d}$, hence applying the Hoeffding inequality and the choice of $q_i= \lceil \frac{2^{2d+1} \log(2/\eta')}{\epsilon_i^2} \rceil$ we obtain
\begin{equation*}
    \Pr\Big[ \Big| \hat{\mu}^{(i)} - \frac{1}{|S_k|}\tr\big(H^i\big)\Big| \geq \frac{\epsilon}{2 \lceil d/2 \rceil|a_i|}\Big] \leq 2\exp\Big(\frac{-\epsilon^2 q_i}{\lceil d/2 \rceil^2|a_i|^2 2^{2d+1}} \Big) \leq \eta^\prime
\end{equation*}
for each $i=1,...,d$. Together we have therefore shown 
\begin{equation*}
\begin{split}
    \Pr\Big[ \Big|\sum_{i=1}^d |a_i|\Big( \hat{\mu}^{(i)} - \frac{1}{|S_k|}\tr\big(H^i\big)\Big)\Big| \leq \frac{\epsilon}{2}\Big] \geq \Big(1 - \eta^\prime \Big)^{\lceil d/2 \rceil} \geq 1- \eta.
\end{split}
\end{equation*}
\end{proof}

\begin{theorem}
\label{thm:complexity_apers}
Let $G$ be a given graph on $n$ vertices, $k$ a desired dimension, $\epsilon$ a desired error, $\eta$ a desired failure probability, and let $\delta$ be the spectral gap of the normalised Laplacian $\tilde{\Delta}_k$. Executing Algorithm \ref{alg:apers_2} gives an $(\epsilon, \eta)$-estimator of $\beta_k/|S_k|$ of the clique complex of $G$, with sample complexity
\begin{equation*}
    S_{\text{\apersTwoName}} =  \Big(\frac{1}{\sqrt{\delta}}\log(4/\epsilon) \Big)^3 2^{8\frac{1}{\sqrt{\delta}}\log(4/\epsilon)-1} \log(2/\eta)
\end{equation*}
and query complexity 
\begin{equation*}
    Q_{\text{\apersTwoName}} = \Big(\frac{1}{\sqrt{\delta}}\log(4/\epsilon) \Big)^4 2^{8\frac{1}{\sqrt{\delta}}\log(4/\epsilon)-1} \log(2/\eta)
\end{equation*}
\end{theorem}

\begin{proof}
    The sample complexity of Algorithm~\ref{alg:apers_2} is $\sum_{i=1}^d q_i$. Applying the choice of $q_i = \lceil \frac{2^{2d+1} \log(2/\eta')}{\epsilon_i^2} \rceil$ and $\epsilon_i = \frac{\epsilon}{2 \lceil d/2 \rceil |a_i|}$ this is simplified as
\begin{equation*}
\begin{split}
    S_{\apersTwoName} &= \sum_{i=1}^d \frac{2^{2d+1} \log(2/\eta^\prime)}{\epsilon_i^2} \\
    &= 2^{2d+1} \log(2/\eta^\prime) \sum_{i=1}^d  \frac{4 \lceil d/2 \rceil^2 |a_i|^2}{\epsilon^2} \\
    &\leq \frac{1}{\epsilon^2} d^2 2^{2d+1} \log(2/\eta^\prime) \|p\|_2^2 \\
\end{split}
\end{equation*}
We then apply the bound on $2$-norm given in Lemma~\ref{lem-coef_bound2}, so that the above is upper bounded as
\begin{equation*}
    S_{\apersTwoName} \leq d^2 2^{8d-1} \log(2/\eta^\prime)
\end{equation*}
Lastly, note that that $\eta'$ is bounded below by $\eta^{\lceil d/2 \rceil}$. This can be seen by using the Taylor expansion of $\eta'$ around $\eta = 0$ and noticing that $\eta' > \eta/\lceil d/2 \rceil$ which is bounded below by $\eta^{\lceil d/2 \rceil}$ for any $\eta<0.5$ and $\lceil d/2 \rceil>2$. This shows the upper bound
\begin{equation*}
    S_{\apersTwoName} \leq d^3 2^{8d-1} \log(4/\eta).
\end{equation*}
Substituting $d = \frac{1}{\sqrt{\delta}}\log(4/\epsilon)$ we get the claimed sample complexity. For the query complexity, for each $1\leq i \leq d$, each of the $q_i$ samples used to estimate $\tr(H^i)/|S_k|$ make $i \leq d$ calls to the matrix $H$. Therefore the total query complexity is bounded as 
\begin{equation*}
    Q_{\apersTwoName} \leq d \times S_{\apersTwoName} \leq d^4 2^{8d-1} \log(4/\eta),
\end{equation*}
completing the proof.
\end{proof}

\section{The QBNE-Power algorithm}\label{sec:power_method}
We now present a new quantum algorithm for Betti number estimation which requires $\mathcal{O}(1/\epsilon^2)$ samples from short-depth quantum circuits. Assessing advantage for this algorithm requires consideration of the variable convergence rate of Algorithm \ref{alg:apers_1}. We compare these algorithms empirically in Section \ref{sec:num_exp}.

As shown in Theorems \ref{thm:complexity_ibm_one} and \ref{thm:complexity_apers}, the number of samples required to estimate the $k^{\textit{th}}$ normalised Betti number of the input graph grows at least exponentially in the term $\log(1/\epsilon)/\sqrt{\delta}$. In this section, we describe a new alternative quantum algorithm for this problem which exponentially improves the Monte Carlo algorithm of Akhalwaya et al.\ studied in Section \ref{sec:nisq-tda} which, as we show in Theorem \ref{thm:complexity_this_paper}, has a sample count that is polynomial in $n, 1/\epsilon$ and $1/\delta$.

\subsection{Outline}
\label{sec:power_method_description}

In this section we propose a new quantum algorithm for Betti number estimation which can be viewed as a quantum analogue of Algorithm \ref{alg:apers_1} in Section~\ref{sec:apers}. The algorithm relies on first modifying the circuit construction of Akhalwaya et al.\ to work for the \emph{reflected} Laplacian $I-\tilde{\Delta}_k$, showing that we can write $I-\tilde{\Delta}_k = D^{\dagger}D$ and giving a block-encoding $U_{D}$. Following the notation of Apers et al., we refer to this matrix as $H$.  Then we use the stochastic trace estimation technique described in Algorithm \ref{alg:block_encoded_trace} to estimate $\tr(H^d)$ which approximates the $k^{\textit{th}}$ normalised Betti number for a sufficiently high $d$. This method is summarised in Algorithm \ref{alg:this_paper}. 

\begin{algorithm}[t]
\caption{A New Quantum Algorithm for Betti Number Estimation}\label{alg:this_paper}
\begin{algorithmic}[1]
\Require {The input is per that of Algorithm \ref{alg:cap}}
\Ensure {The output is per that of Algorithm \ref{alg:cap}}
\Let $\normalLap, R_k$ and $\delta$ be defined as per Algorithm \ref{alg:cap}.
\Procedure{\thisPaperName}{$\mathcal{G}, k, \epsilon$}
\State $d \gets \lceil \frac{\log(2/\epsilon)}{\delta} \rceil $ 
\Comment{Apply Lemma \ref{lem-power_method_exp} for $\epsilon/2$.}
\State $q \gets \lceil \frac{2\log(2/\eta)}{\epsilon^2} \rceil $ \Comment{See Theorem \ref{bettihat_algorithm}.}

\State Let $U_D$ be the block encoding of a matrix $D$ such that $D^\dagger D = I - \normalLap$
\State \Return $\Call{EstimateFromBlockEncoding}{U_D, d, q, R_k}$ 
\EndProcedure
\end{algorithmic}
\end{algorithm}

\subsection{Quantum circuits for the reflected Laplacian}
\label{sec:H_blockencoding}
In Section \ref{sec:block_encoding}, we recalled the quantum circuits designed by Akhalwaya et al.\ for constructing trace estimates of powers of the normalised Laplacian matrix $\tilde{\Delta}_k$. Central to this,(\ref{eqn:decomp}) gives a modular decomposition of the Laplacian into components which could be implemented in quantum circuits as unitaries or block-encodings. However, taking a circuit implementing some unitary $U$ and trying to design a circuit implementing $I-U$ is not even possible in general. Fortunately, the Laplacian has structure that allows us to give a decomposition of the reflected normalised Laplacian, $I-\tilde{\Delta}_k$ in (\ref{eqn:H-decomp}), which differs by just one component to that in (\ref{eqn:decomp}):
\begin{equation}\label{eqn:H-decomp}
   H = P_kP_\Gamma \Big(\frac{1}{\sqrt{n}}B\Big) (I-P_\Gamma) \Big(\frac{1}{\sqrt{n}}B \Big) P_{\Gamma}P_k.
\end{equation}
By this construction, $H$ can be expressed as $D^{\dagger}D$ where 
$$D = (I-P_\Gamma) \Big(\frac{1}{\sqrt{n}}B \Big) P_{\Gamma}P_k.$$
To show how to compute the block-encoding $U_D$ required to apply Algorithm \ref{alg:block_encoded_trace} to $H$ it remains to show how to block encode the projection $I-P_{\Gamma}$. Here, we show that this can be done.

This circuit, $U_{I-P_\Gamma}$, is shown in Figure \ref{fig:fix_circuit}. It works by first applying a circuit $U_{P_\Gamma}$ which block-encodes the projection $P_\Gamma$, then applying a multi-controlled $X$ controlled on the $0$ outcome of every one of the auxiliary qubits of $U_{P_\Gamma}$. Finally, we apply $U_{P_\Gamma}^{\dagger}$ to uncompute the auxiliary qubits. See \cite{akhalwaya2024topological} for a circuit construction of $U_{P_\Gamma}$ using $\mathcal{O}(n^2)$ auxiliary qubits.

\begin{theorem}
    Given a circuit $U_{P_\Gamma}$ which block-encodes the projection $P_\Gamma$, the circuit in Figure~\ref{fig:fix_circuit} acts as a block-encoding of the projection $I - P_\Gamma$. 
\end{theorem}

\begin{proof}
    Note that the projection $I - P_\Gamma$ sends a computational basis state to itself if the corresponding simplex is not in $\Gamma$, and $0$ otherwise. We will show that the circuit maps a given computational basis state $\ket{x_1,...,x_n}\ket{0^m}\ket{0}$ to itself when the set corresponding to $\ket{x_1,...,x_n}$ is not in $\Gamma$, and otherwise maps to $\ket{x_1,...,x_n}\ket{\psi}\ket{0}$ for some state $\ket{\psi}$ orthogonal to $\ket{0^m}\ket{0}$. If the set corresponding to $\ket{x_1,...,x_n}$ is not in $\Gamma$, then the circuit $U_{P_\Gamma}$ acts on $\ket{x_1,...,x_n}\ket{0^m}\ket{0}$ as the identity, after which the following two operations maps it as
\begin{equation*}
\begin{split}
    \ket{x_1,...,x_n}\ket{0^m}\ket{0} &\mapsto \ket{x_1,...,x_n}\ket{0^m}\ket{1} \\
    &\mapsto \ket{x_1,...,x_n}\ket{0^m}\ket{0} \\
\end{split}
\end{equation*}
    The remaining $U_{P_\Gamma}^\dagger$ acts trivially on this output state since the set corresponding to $\ket{x_1,...,x_n}$ is not in $\Gamma$. This shows that the circuit has the correct action when the set corresponding to $\ket{x_1,...,x_n}$ is not in $\Gamma$. For case when this set is in $\Gamma$, the circuit $U_{P_\Gamma}$ maps $\ket{x_1,...,x_n}\ket{0^m}\ket{0}$ to $\ket{x_1,...,x_n}\ket{\psi}\ket{0}$ for some state $\ket{\psi}$ orthogonal to $\ket{0^m}$. The multi-controlled $X$ operation then acts trivially, and the $X$ gate maps the state to $\ket{x_1,...,x_n}\ket{\psi}\ket{1}$. After applying the circuit $U_{P_\Gamma}^\dagger$ the output is $\ket{x_1,...,x_n}\ket{0^m}\ket{1}$, where $\ket{0^m}\ket{1}$ is indeed orthogonal to $\ket{0^m}\ket{1}$. This completes the proof.
\end{proof}

\begin{figure}
\centering
    \begin{tikzpicture}
    \begin{yquant}
    qubit {$\ket{0}$} target;
    qubit {$\ket{0}$} a[1];
    qubit {$\vdots$} a[+1];
    qubit {$\ket{0}$} a[+1];
    discard a[1];
    qubit {$\ket{v_{\protect\the\numexpr\idx+1}}$} q[4]; 
    box {$U_{P_\Gamma}$} (q, a);
    barrier target, a, q;
    cnot  target | ~ a;
    box {$X$} target;
    barrier target, a, q;
    box {$U_{P_\Gamma}^{\dagger}$} (q, a);
    \end{yquant}
    \end{tikzpicture}
\caption{The circuit $U_{I-P_\Gamma}$ which block-encodes the projection $I-P_{\Gamma}$ given a circuit $U_{P_\Gamma}$ that block-encodes the projection $P_\Gamma$. The multi-controlled $X$ operation is controlled on the $0$ of each of the auxiliary qubits of $U_{P_\Gamma}$. See \cite{akhalwaya2024topological} for a circuit implementation of $U_{P_\Gamma}$ using $\mathcal{O}(n^2)$ auxiliary qubits.}
\label{fig:fix_circuit}
\end{figure}

\subsection{Complexity analysis}

In this section, we prove the correctness of Algorithm \ref{alg:this_paper} by verifying the number of samples $q$ given. First, in Theorem \ref{bettihat_algorithm} we establish the correctness of the estimator created by Algorithm \ref{alg:this_paper}. Then we show in Theorem \ref{thm:complexity_this_paper} that both the sample complexity and query complexity grow only polynomially in $n, 1/\epsilon$, and $1/\delta$. This represents a large asymptotic improvement over the behaviour of the previously presented classical and quantum algorithms.    

The algorithm presented in Algorithm~\ref{alg:this_paper} produces an estimate for $\beta_k/|S_k|$ by estimating the $\tr(H^d)/|S_k|$ using Algorithm \ref{alg:block_encoded_trace}. We recall that the algorithm generates samples which are either $0$ or $1$ and has expectation $\tr(H^d)/|S_k|$. To establish the correctness of Algorithm \ref{alg:this_paper}, we prove that the number of samples that we pass to this trace estimation subroutine is sufficient. That is the purpose of the next result.

\begin{theorem}
\label{bettihat_algorithm}
Let $\epsilon, \eta>0$. For all $q \geq 2\log(2/\eta)/\epsilon^2$ and $d \geq \frac{\log(2/\epsilon)}{\delta}$ the Betti number estimator $\hat{\beta}_k$ output by Algorithm \ref{alg:this_paper} is an $(\epsilon, \eta)$-estimator.
\end{theorem}

\begin{proof}
The proof is similar to that of Theorem~\ref{thm-IBM_betti_estimate} and so we sketch the main ideas. Similar to Theorem~\ref{thm-IBM_betti_estimate}, the normalised Betti number estimate provided by Algorithm~\ref{alg:this_paper} can be expressed as an average of random variables taking value $0$ or $1$, say $\hat{\beta}_k = \frac{1}{q}\sum_{i=1}^q \hat{\mu}_i$. Applying a similar triangle inequality along with Lemma~\ref{lem-power_method_exp} we obtain the lower bound
\begin{equation*}
    \Pr\Big[ \Big|\hat{\beta}_k - \frac{\beta_k}{|S_k|}\Big| \leq \epsilon \Big] \geq \Pr\Big[ \Big|\frac{1}{q}\sum_{i=1}^q \hat{\mu}_i - \frac{1}{|S_k|} \tr(H^d) \Big| \leq \epsilon / 2 \Big].
\end{equation*}
Applying the Hoeffding inequality of Lemma~\ref{thm-hoeffding}, we obtain
\begin{equation*}
   \Pr\Big[ \Big|\frac{1}{q}\sum_{i=1}^q \hat{\mu}_i - \frac{1}{|S_k|} \tr(H^d) \Big| \geq \epsilon / 2 \Big] \leq 2\exp(-\epsilon^2 q/2),
\end{equation*}
and therefore the choice of $q$ gives the result.
\end{proof}

We can now summarise the quantum algorithm for Betti number approximation using the power method and provide its time complexity. 

\begin{theorem}\label{thm:complexity_this_paper}
Let $G$ be a given graph on $n$ vertices, $k$ a desired dimension, $\epsilon$ a desired error, $\eta$ a desired failure probability, and let $\delta$ be the spectral gap of the normalized Laplacian $\tilde{\Delta}_k$. Executing Algorithm \ref{alg:this_paper} gives us an $(\epsilon, \eta)$-estimator of $\beta_k/|S_k|$ of the clique complex of $G$, with sample complexity 
\begin{equation*}
    S_{\text{\thisPaperName}} = \frac{2\log(2/\eta)}{\epsilon^2} 
\end{equation*}
and query complexity
\begin{equation*}
    Q_{\text{\thisPaperName}} = \frac{\log(2/\epsilon)}{\delta} \times \frac{2\log(2/\eta)}{\epsilon^2}
\end{equation*}
\end{theorem}

\begin{proof}
Following Theorem~\ref{bettihat_algorithm}, the total number of samples used is $\frac{2\log(2/\eta)}{\epsilon^2}$, which gives the claimed sample complexity. For the query complexity, each sample produced by Algorithm \ref{alg:this_paper} requires running a circuit with at most $d = \frac{1}{\delta} \log(2/\epsilon)$ calls to the the circuit $U_D$ or $U_{D^{\dagger}}$. Thus the total number of uses of this circuit is given as 
\[  \frac{\log(2/\epsilon)}{\delta} \times \frac{2\log(2/\eta)}{\epsilon^2} ,\]
as required.
\end{proof}

\section{Numerical experiments}\label{sec:num_exp}

As summarized in Table \ref{tab:comparison}, we have analysed four Monte Carlo algorithms for normalised Betti number estimation, deriving upper bounds for the required number of samples to achieve an $(\epsilon, \eta)$-estimate for a given error $\epsilon$ and confidence $\eta$. In this section, we set out to empirically verify our analysis by implementing all four algorithms and confirming that the output of the algorithms converge to the known ground-truth values within a required precision, $\epsilon$, using a number of samples less than or equal to the conservative upper bounds. We also set out to observe empirical performance differences between the four algorithms.

The number of samples needed by the four algorithms to produce the estimate $\hat{\beta}_k$ for a user-selected order $k$ naturally depends on the two user-provided `output-quality' parameters $\epsilon$ and $\eta$, which we choose as $\epsilon=\eta=0.1$ throughout. More opaquely, the sample counts depend on subtle properties of the user-provided graph, most importantly, the spectral gap $\delta$ of the normalised combinatorial Laplacian. 
The number of vertices, $n$, indirectly features in the upper bound on the number of samples through the spectral gap, where as $n$ increases the gap may decrease e.g. $\delta \in O(1/\textrm{poly}(n))$ and indeed does for our chosen class of benchmark graphs. The number of vertices also features in the computational time needed to generate one sample.


\subsection{Selected benchmarks: complete ($k+1$)-partite graphs}

We have selected to run the algorithms on four clique complexes which have explicit expressions for both their normalised Betti numbers and spectral gaps of their normalised Laplacian, and which have been previously studied in \cite{Berry2024}.\footnote{These graphs in graphml format can be found at \url{https://github.com/quantinuum-dev/CBNE/tree/main/graphs}} For this reason we consider the complete $(k+1)$-partite graph, which is the graph having $k+1$ clusters with $m$ vertices in each cluster, such that any two distinct vertices are adjacent if they are in different clusters. The clique complex of this graph has normalised Betti number $\beta_k/|S_k| = (m-1)^{k+1}/m^k$, and the spectral gap of $\tilde{\Delta}_k$ is $\delta = \frac{1}{k+1}$ \cite[Proposition 1,2]{Berry2024}. We collect the exactly calculated instances of these properties for the graphs under study in Table \ref{tab:graph_properties}. These graphs are useful benchmarks because they have a small number of vertices and yet their sample complexities as described in Table~\ref{tab:comparison} are large. Additionally, the graphs that induce these clique complexes are the smallest non-trivial examples of the class of graphs discussed in \cite{Berry2024} that induce exponentially large Betti numbers. In Table \ref{tab:step_count_comparison}, we list the degree $d$ of the respective polynomial used, along with the sample and query counts for these cases, which are computed by the formulae referenced in Table~\ref{tab:comparison}.

\begin{table}[!tb]
    \centering
    \resizebox{\textwidth}{!}{%
    \setlength{\tabcolsep}{10pt}
    \renewcommand{\arraystretch}{1.2}
    \begin{tabular}{c >{\Centering\arraybackslash}m{2cm}  >{\Centering\arraybackslash}m{2cm} >{\Centering\arraybackslash}m{2cm} >{\Centering\arraybackslash}m{2cm}
    >{\Centering\arraybackslash}m{2cm}}
    \cline{2-6}
          & Dimension $k$ & Number of Vertices, $n$ {\small $= (k+1)\times m$}& Spectral gap $\delta$ &  Normalised Betti number $\beta_k / |S_k|$ & Layout\\\hline
    \multicolumn{1}{c|}{Graph-1} & 1 & $6=2\times 3$ & 0.500  & 0.444 & \vspace{0.5em} \includegraphics[width=0.1\textwidth, height=15mm]{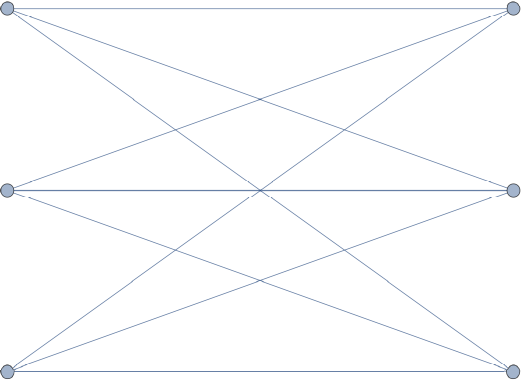} \\ \hline
    \multicolumn{1}{c|}{Graph-2} & 1 & $8=2\times 4$ & 0.500  & 0.562 & \vspace{0.5em} \includegraphics[width=0.1\textwidth, height=15mm]{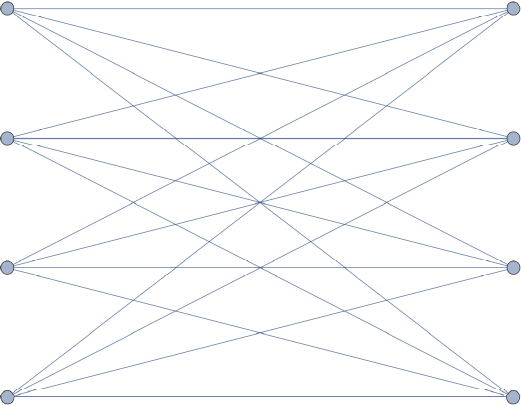}\\ \hline
    \multicolumn{1}{c|}{Graph-3} & 2 & $9=3\times 3$ & 0.333  & 0.296 & \vspace{0.5em} \includegraphics[width=0.1\textwidth, height=15mm]{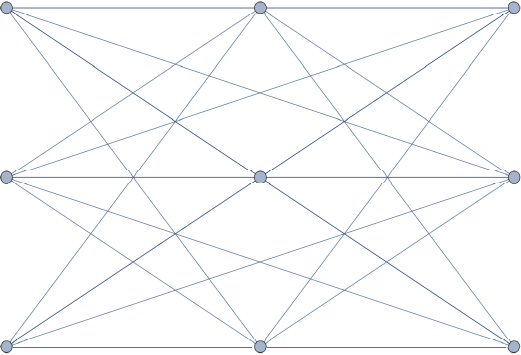} \\ \hline
    \multicolumn{1}{c|}{Graph-4} & 2 & $12=3\times 4$ & 0.333 & 0.421 &  \vspace{0.5em} \includegraphics[width=0.1\textwidth, height=15mm]{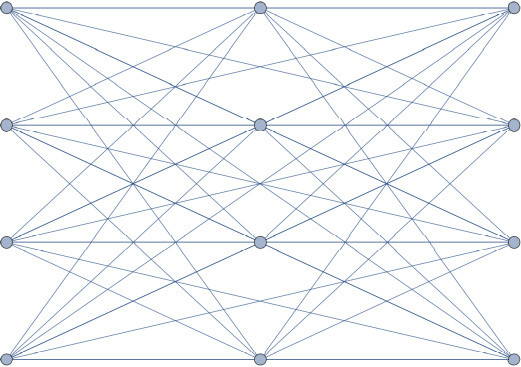} \\ \hline
    \end{tabular}}
    \caption{Properties of the four $k+1$-partite graphs with $m$ clusters used to compare the algorithms.}
    \label{tab:graph_properties}
\end{table}


\begin{table}[!tb]
    \centering
    \setlength{\tabcolsep}{10pt}
    \renewcommand{\arraystretch}{1.2}
    \resizebox{\textwidth}{!}{%
    \begin{tabular}{c | c 
    >{\Centering\arraybackslash}m{1.5cm}
    >{\Centering\arraybackslash}m{1.5cm} | c  
    >{\Centering\arraybackslash}m{1.5cm}
    >{\Centering\arraybackslash}m{1.5cm} | c 
    >{\Centering\arraybackslash}m{1.5cm}
    >{\Centering\arraybackslash}m{1.5cm} | c 
    >{\Centering\arraybackslash}m{1.5cm}
    >{\Centering\arraybackslash}m{1.5cm} |} 
    \cline{2-13}  
     & \multicolumn{3}{c|}{\ibmOneName} & \multicolumn{3}{c|}{\apersOneName} & \multicolumn{3}{c|}{\apersTwoName} & \multicolumn{3}{c|}{\thisPaperName} \\\cline{2-13}  
      & $d$ & Sample Count & Query Count & $d$ & Sample Count & Query Count & $d$ & Sample Count & Query Count & $d$ & Sample Count & Query Count \\\hline
    \multicolumn{1}{|c|}{Graph-1} & 6 & $2.07 \times 10^{19}$ & $8.70 \times 10^{20}$ & 6 & $2.45 \times 10^6$ & $1.47 \times10^7$ & 6 & $9.11 \times 10^{16}$  &$5.46 \times 10^{17}$ & 6 & $6.00 \times 10^2$ & $3.60 \times 10^3$ \\ \hline
    \multicolumn{1}{|c|}{Graph-2} & 6 & $2.07 \times 10^{19}$ & $8.70 \times 10^{20}$ & 6 & $2.45 \times 10^6$ &$1.47 \times 10^7$ & 6 & $9.11 \times 10^{16}$ & $5.46 \times 10^{17}$ & 6 & $6.00 \times 10^2$ &$3.60 \times 10^3$\\ \hline
    \multicolumn{1}{|c|}{Graph-3} & 7 & $2.48 \times 10^{22}$ &$1.39 \times 10^{24}$ & 9 & $1.57 \times 10^8$ &$1.41 \times 10^9$ & 7 & $3.70 \times 10^{19}$ &$2.59 \times 10^{20}$ & 9 & $6.00 \times 10^2$ & $5.40 \times 10^3$ \\ \hline
    \multicolumn{1}{|c|}{Graph-4} &  7 & $2.48 \times 10^{22}$ &$1.39 \times 10^{24}$ & 9 & $1.57 \times 10^8$ &$1.41 \times 10^9$ & 7 & $3.70 \times 10^{19}$ & $2.59 \times 10^{20}$ & 9 & $6.00 \times 10^2$ &$5.40 \times 10^3$ \\ \hline
    \end{tabular}}
    \caption{Comparison of the sample counts and query counts for the four different algorithms to estimate the normalised Betti number $\beta_k/|S_k|$ with error $\epsilon=0.1$ and failure probability $\eta = 0.1$. See Table \ref{tab:graph_properties} for the graph properties.}
    \label{tab:step_count_comparison}
\end{table}

\subsection{Classical implementations of the four Monte Carlo algorithms}

We have implemented the two classical algorithms in C++, closely following the descriptions in Algorithms \ref{alg:apers_1} and \ref{alg:apers_2} while incorporating the improvements introduced in this paper.\footnote{Our implementation can be found at \url{https://github.com/quantinuum-dev/CBNE/tree/main}. The repository contains instructions on building and running the tool.} 
For the two quantum algorithms, since we are mainly focusing on comparing the sample count behaviour in the noiseless regime with as large a vertex count as manageably possible, we decided against implementing them on a quantum computer or even using a quantum programming language, preferring to classically simulate the unitary and projection matrices acting on the simplicial subspace of the $\ket{0^a}$-block only, i.e. directly simulating the (otherwise block-encoded) $D$ and $D^\dagger$ on $S_k$ represented with $n$ qubits, using a symbolic algebra package. The most important reason for this is to avoid simulating the full Hilbert space of $n+a$ qubits, thereby achieving an exponential classical simulation saving. The unitary matrices acting solely on the main simplex register are calculated by directly simulating actual gates acting on the main register qubits. However, given our strategy to avoid simulating the full Hilbert space, we have to forgo empirically checking the correctness of the individual quantum gates acting on the auxiliary qubits, satisfying ourselves with mathematically equivalent operations. In particular, the control gates targeting the auxiliary qubits followed by mid-circuit measurement of the auxiliary qubits (with the concomitant state collapse of the main register) are simulated by the following procedure on the main register only. We implement the Markov steps by applying the Hermitian matrices $D$ or $D^\dagger$, calculated by sandwiching a circuit-derived unitary matrix with circuit-equivalent projection matrices. To simulate quantum state collapse and the generation of the measurement outcomes, we draw a uniform random number in $[0,1]$ and compare it to the value of the norm of the non-normalized simplicial state vector in the block. If the random number is less than the value of the norm, we manually normalize the simplex state vector thereby simulating a successful projection onto $\ket{0^a}$, and continue with the remaining Markov steps. If the projection `fails', we record a zero for $s_i$ and move to the next sample. If all $d$ Markov steps end with successful projections, we record a one for $s_i$ and move to the next sample. With this randomised procedure we are able to accurately and realistically simulate lines 10 - 15 of Algorithm \ref{alg:block_encoded_trace}.

\begin{table}[!tb]
    \centering
    \setlength{\tabcolsep}{10pt}
    \renewcommand{\arraystretch}{1.2}
    \begin{tabular}{
    m{1.5cm} >{\Centering\arraybackslash}
    m{0.1cm} >{\Centering\arraybackslash}
    m{2cm}   >{\Centering\arraybackslash}
    m{2cm}   >{\Centering\arraybackslash}
    m{2cm}   >{\Centering\arraybackslash}
    m{2cm}}
     & &\ibmOneName & \apersOneName  & \apersTwoName & \thisPaperName \\
     Graph-1 & 
     $\hat{\beta}_1$&\includegraphics[width=0.15\textwidth, height=20mm]{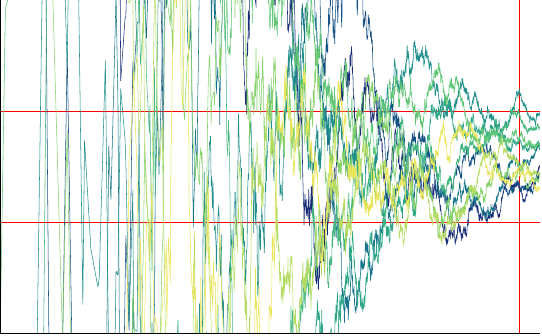}  & \includegraphics[width=0.15\textwidth, height=20mm]{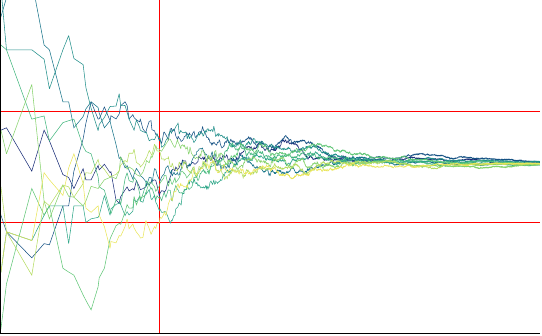} & \includegraphics[width=0.15\textwidth, height=20mm]{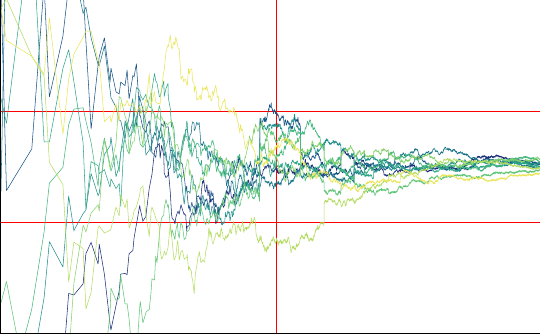}  & 
     \includegraphics[width=0.15\textwidth, height=20mm]{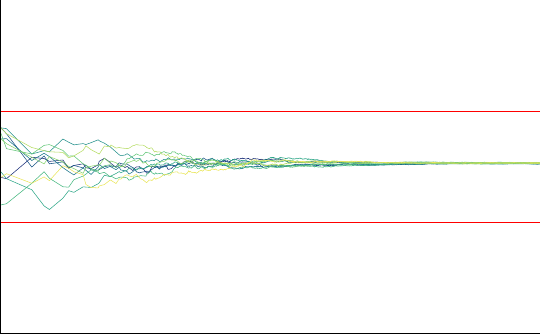} \\ 
    Graph-2 & 
    $\hat{\beta}_2$&\includegraphics[width=0.15\textwidth, height=20mm]{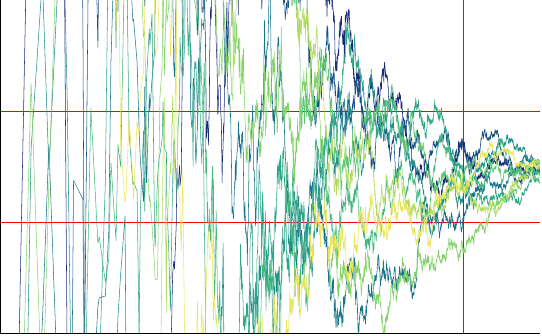} & \includegraphics[width=0.15\textwidth, height=20mm]{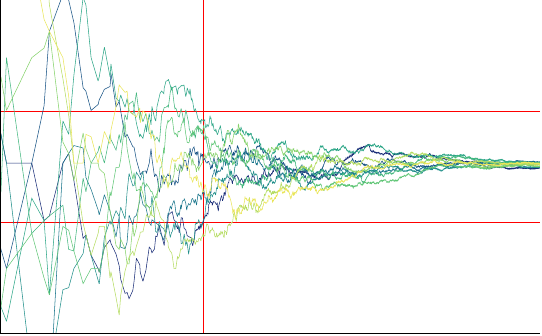} & \includegraphics[width=0.15\textwidth, height=20mm]{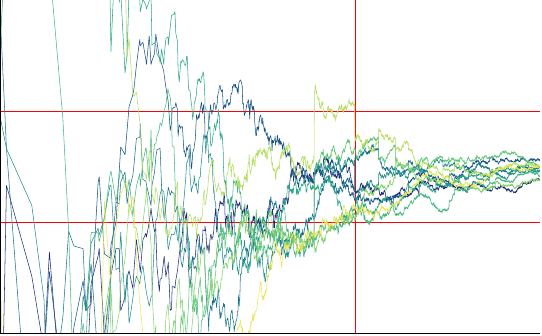}  & 
    \includegraphics[width=0.15\textwidth, height=20mm]{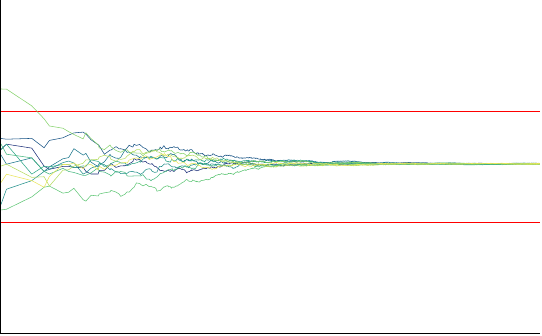} \\ 
     Graph-3 & 
     $\hat{\beta}_1$&\includegraphics[width=0.15\textwidth, height=20mm]{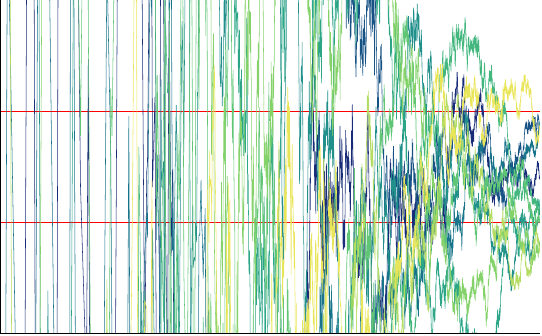}  & \includegraphics[width=0.15\textwidth, height=20mm]{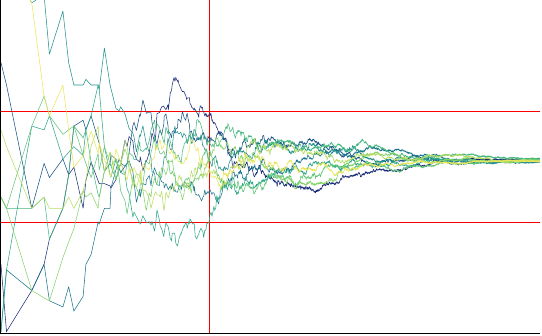} & \includegraphics[width=0.15\textwidth, height=20mm]{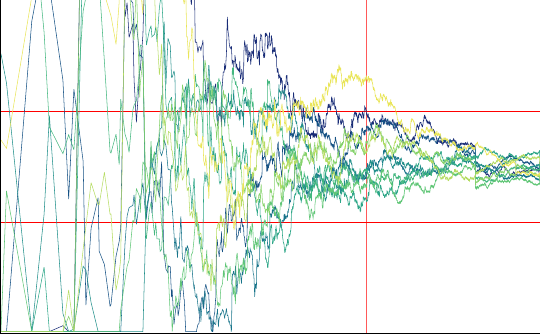} & 
     \includegraphics[width=0.15\textwidth, height=20mm]{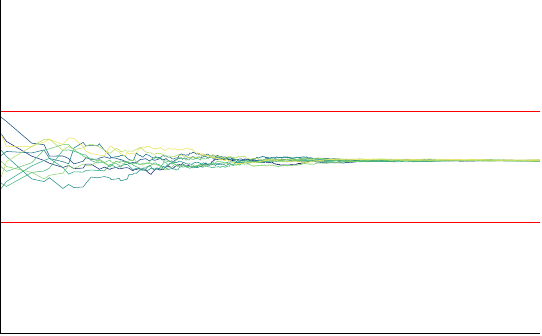} \\ 
     Graph-4 & 
     $\hat{\beta}_2$&\includegraphics[width=0.15\textwidth, height=20mm]{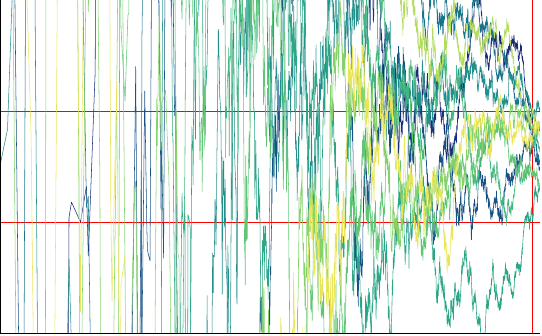}  & \includegraphics[width=0.15\textwidth, height=20mm]{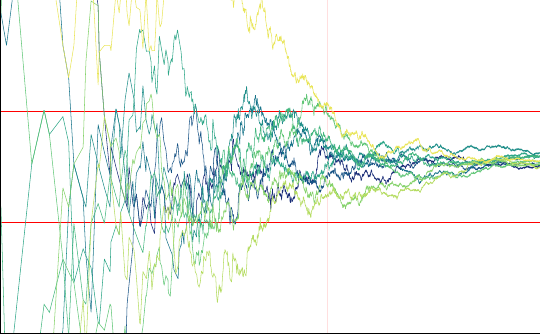}  & \includegraphics[width=0.15\textwidth, height=20mm]{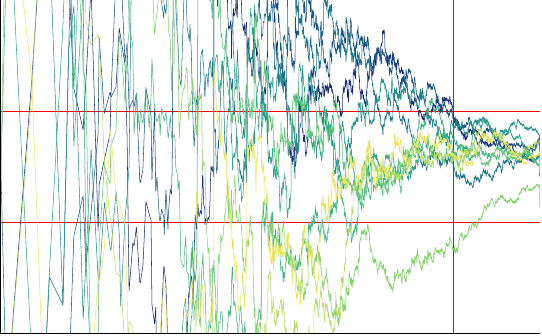}  & \includegraphics[width=0.15\textwidth, height=20mm]{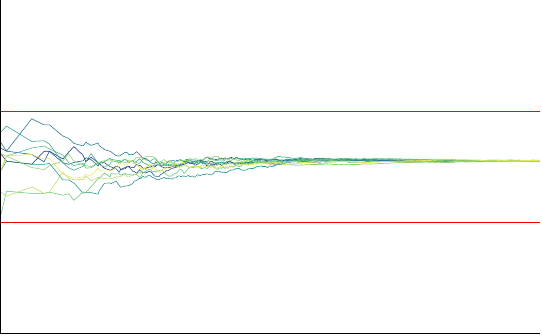}  \\ 
    &&\sampleaxis & \sampleaxis & \sampleaxis & \sampleaxis\\
    \end{tabular}
    \caption{$\hat{\beta}_k$ vs sample count for 10 runs of the four algorithms on each of the four benchmark graphs. Horizontal lines correspond to the desired $\epsilon$ precision interval. Vertical line corresponds to the first count where 9 out of 10 runs empirically remain within $\epsilon$ of the ground-truth (not visible when out of range).}
    \label{tab:run_graphs}
\end{table}

\subsection{Results and interpretation of experiments}
Having explained the benchmark graphs and the implementation of the algorithms, we now discuss the experiments we ran, the results we obtained and our interpretation of the results. As discussed above, we decided to compare the convergence of the four algorithms on the four benchmark graphs as a function of the number of samples taken and not the time taken, since there is an uninteresting time dependence on the different algorithms and their implementations to produce each sample. We also decided to run 10 instances of each of these 16 experiments to allow us to observe that indeed there is significant variation between runs and that our analysis accurately captures this.

With four different algorithms on four benchmark graphs we have widely varying sample counts, as described in Table~\ref{tab:step_count_comparison}. However, in order to facilitate a straightforward comparison between the algorithms as well as between graphs for the same algorithm we have decided to run all experiments for the same fixed sample count of $10^6$ samples. For algorithms other than \thisPaperName{} this sample count is well below the counts required to guarantee convergence as per Table \ref{tab:step_count_comparison}. 

We display the resulting 16 plots in a grid of graphs versus algorithms in Table \ref{tab:run_graphs}. Each plot is the running estimate of the normalised Betti number plotted against the sample count on a log-scale starting at $10^2$, in equal logarithmic steps up until $10^6$ samples. The two horizontal red lines represent the chosen error $\epsilon$ from the true normalised Betti estimate (which is subtracted from the 10 traces to center the plot). The vertical red lines correspond to the first sample count after which 9 out of 10 of the runs remain within $\epsilon$ of the actual normalized Betti number. Fortunately, as expected, convergence occurs much earlier than the sample counts computed in Table~\ref{tab:step_count_comparison}.

Overall the most striking difference between the four algorithms as illustrated in Table~\ref{tab:run_graphs} is the performance of \thisPaperName{} compared to the other three. The estimates output by \thisPaperName{} in Table~\ref{tab:run_graphs} appear to have the smallest variance from the actual Betti number, and produces an estimate within $\epsilon$ using far fewer samples than the other three algorithms. This observation is justified by the sample complexity of \thisPaperName{} only depending on $\epsilon$ and $\eta$, as displayed in Table~\ref{tab:comparison}.

\section{Conclusion}
We have studied four Monte Carlo algorithms for Betti number estimation. Our analysis of the three algorithms already found in the literature improves previous understanding. Furthermore, we introduce a new quantum algorithm that does not suffer from an exponential dependence on the Laplacian inverse-eigengap. The emerging picture is that both quantum approaches benefit from exploring exponentially-many Monte Carlo paths in one circuit run, while both are sample-noise limited by the number of samples needed to extract information about the powers. We have shown that by using the reflected Laplacian, it becomes possible to avoid the exponential precision needed in the power estimation, thereby avoiding the exponential dependence on $1/\delta$ which is present in the previous algorithms. For future work, it would be interesting to determine if other concentration inequalities aside from the Hoeffding inequality could lead to tighter upper bounds.

\newpage
\appendix

\section{Proofs}
Here, we present the proofs for the Lemmas presented in them main section of the paper.
\subsection{Proof of Lemma~\ref{lem:quantum_trace_est_alg_correctness}}\label{app:L3.3}
\begin{lemma}[Restatement of Lemma~\ref{lem:quantum_trace_est_alg_correctness}]
    For each $i=1,...,q$, the random variable $s_i$ output by Algorithm~\ref{alg:block_encoded_trace} is a Bernoulli random variable with expectation $\frac{1}{N}\tr(M^d)$. 
\end{lemma}
\begin{proof}
Clearly each random variable $s_i$ takes value $0$ or $1$ by definition, so we are left to show that the probability that $s_i=1$ is $\frac{1}{N}\tr(M^d)$. From Lemma~\ref{lem:psd_mat_decomp} we know that the the normalized trace of $M^d$ can be expressed as $\frac{1}{N}\tr(M^d) = \frac{1}{N} \sum_{x} \| D^{(d)}\ket{x} \|^2$ where $D^{(d)}$ is a product of $d$ terms equal to $D$ or $D^\dagger$. The random variable $s_i$ is determined by choosing a bitstring $x$ uniformly from the set $S$ and then successively applying the block-encoding of $D$ or its Hermitian conjugate, hence the expected value of $s_i$ is 
\begin{equation*}
    \sum_{x \in S} \frac{1}{N} \Pr(\textnormal{all measurements are \ } 0^a),
\end{equation*}
where the above probability refers to the measurement outcomes described in Algorithm~\ref{alg:block_encoded_trace} line 10 each being $0^a$. We now show that this probability for a given $x \in S$ is $\| D^{(d)}\ket{x} \|^2$ by induction on the degree $d$. 

For the case where $d=1$, Algorithm~\ref{alg:block_encoded_trace} prepares the state 
\begin{equation*}
    U_D \ket{0^a}\ket{x} = \ket{0^a}D\ket{x} + \ket{\psi}
\end{equation*}
where $\ket{\psi}$ is orthogonal to $\ket{0^a}\ket{y}$ for all $y \in S$. After measuring the auxiliary qubits if we observe the outcome $0^a$ the resulting state is
\begin{equation*}
    \frac{1}{\|D\ket{x}\|} \ket{0^a}D\ket{x}
\end{equation*}  
and the probability that we observe this outcome is $\|D\ket{x}\|^2$, which proves the case $d=1$. Next, assume that the statement is true for all degrees $i < d$ and consider the degree $d$ case, meaning that we have constructed the state
\begin{equation*}
    \frac{1}{\|D^{(d-1)}\ket{x}\|} \ket{0^a}D^{(d-1)}\ket{x}
\end{equation*}
with probability $\|D^{(d-1)}\ket{x}\|^{2}$. The algorithm then applies the circuit $U_{D}$ to this state if $d$ is odd and $U_{D^{\dagger}} = (U_D)^{\dagger}$ if $d$ is even, thus creating the state
\begin{equation*}
    \frac{1}{\|D^{(d-1)}\ket{x}\|}\ket{0^a}D^{(d)}\ket{x} + \ket{\psi}
\end{equation*}
for some state $\ket{\psi}$ that is orthogonal to $\ket{0^a}\ket{y}$ for all $y \in S$. Now if we measure the $a$ auxiliary qubits we observe $0^a$ with probability $\|D^{(d)}\ket{x}\|^2/\|D^{(d-1)}\ket{x}\|^2$. Now the whole process succeeds with probability
\begin{equation*}
    \|D^{(d-1)}\ket{x}\|^2\times \frac{\|D^{(d)}\ket{x}\|^2}{\|D^{(d-1)}\ket{x}\|^2} = \|D^{(d)}\ket{x}\|^2
\end{equation*}
as required.
\end{proof}

\subsection{Proof of Lemma~\ref{lem-coef_bound}}\label{app:L4.3}
We now provide the proof of Lemma~\ref{lem-coef_bound}. Our proof makes use of the following known result regarding the Chebyshev polynomials \cite[Theorem 7]{Erdos1947}.

\begin{lemma}
\label{lem-cheby-max}
Suppose that $q(x)$ is a polynomial of degree at most $d$ with the property that $|q(x)| \leq 1$ for all $x \in [-1, 1]$. Then
\[ |T_d(y)| \geq |q(y)| \]
for every $y \in [-1, 1]^c$.
\end{lemma}

\begin{lemma}[Restatement of Lemma~\ref{lem-coef_bound}]
Let $\epsilon > 0$ and $\delta \in (0, 1/2]$ be given, and let $p(x)$ denote the polynomial 
\begin{equation}
    p(x)=T_d\Big(\frac{1-x}{1-\delta}\Big)/T_d\Big(\frac{1}{1-\delta}\Big).
\end{equation}
For $d \geq \frac{1}{\sqrt{\delta}} \log(2/\epsilon)$, the polynomial $p(x)$ has $2$-norm which satisfies
\begin{equation*}
  \frac{1}{d+1}(9/4)^d \leq \|p\|_2^2 \leq \epsilon^2 d 2^{10d}.
\end{equation*}
\end{lemma}

\begin{proof}
Write $p(x)$ as $a_0 + a_1 x + \ldots + a_d x^d$. For the upper bound, we can derive an explicit expression for the coefficients of $p(x)$ using the well-known identity 
\begin{equation*}
    T_d(x) = d \sum_{k=0}^d (-2)^k \frac{(d+k-1)!}{(d-k)!(2k)!}(1-x)^k.
\end{equation*}
Applying the Binomial Theorem and rearranging we then obtain
\begin{equation}
\label{eqn:cheby_identity}
\begin{split}
    T_d\Big(\frac{1-x}{1-\delta}\Big) &= d \sum_{k=0}^d \Big(\frac{-2}{1-\delta}\Big)^k\frac{(d+k-1)!}{(d-k)!(2k)!}\sum_{i=0}^k \binom{k}{i}(-\delta)^{k-i} x^i \\
    &= d \sum_{i=0}^d \Big[ \sum_{k=i}^d \Big(\frac{-2}{1-\delta}\Big)^k\frac{(d+k-1)!}{(d-k)!(2k)!} \binom{k}{i}(-\delta)^{k-i} \Big]x^i. \\
\end{split}
\end{equation}
From Lemma~\ref{lem-chebypoly}, our choice of $d \geq \frac{1}{\sqrt{\delta}} \log(2/\epsilon)$ implies $1/T_d\big(\frac{1}{1-\delta}\big) \leq \epsilon$, hence the $2$-norm of $p(x)$ satisfies
\begin{equation*}
\begin{split}
    \|p\|_2^2 &\leq \epsilon^2 d^2 \sum_{i=0}^d \Big[ \sum_{k=i}^d \Big(\frac{-2}{1-\delta}\Big)^k\frac{(d+k-1)!}{(d-k)!(2k)!} \binom{k}{i}(-\delta)^{k-i} \Big]^2 \\
    &\leq \epsilon^2 d^3 \sum_{i=0}^d \sum_{k=i}^d \Big(\frac{2}{1-\delta}\Big)^{2k} \Big(\frac{(d+k-1)!}{(d-k)!(2k)!}\Big)^2 \binom{k}{i}^2 \delta^{2(k-i)}  \\
\end{split}
\end{equation*}
where the second inequality follows from the Cauchy-Schwartz inequality applied to the inner sum. We may express $\frac{(d+k-1)!}{(d-k)!(2k)!}$ as the scaled binomial coefficient $\frac{1}{d+k} \binom{d+k}{2k}$. Applying this identity and the fact that $\delta \leq 1$ we obtain
\begin{equation*}
\begin{split}
    \|p\|_2^2 &\leq \epsilon^2 d^3 \Big(\frac{2}{1-\delta}\Big)^{2d} \sum_{i=0}^d \sum_{k=i}^d \frac{1}{(d+k)^2} \binom{d+k}{2k}^2 \binom{k}{i}^2  \\
\end{split}
\end{equation*}
The binomial coefficients $\binom{d+k}{2k}$ are trivially upper bounded by $2^{2d}$ for each $k \leq d$, and additionally $\frac{1}{(d+k)^2} \leq \frac{1}{d^2}$ for $0 \leq k \leq d$, therefore applying this above we obtain
\begin{equation*}
\begin{split}
    \|p\|_2^2 &\leq \epsilon^2 d \Big(\frac{2}{1-\delta}\Big)^{2d} 2^{4d} \sum_{i=0}^d \sum_{k=i}^d \binom{k}{i}^2 \\
    &= \epsilon^2 d \Big(\frac{2}{1-\delta}\Big)^{2d} 2^{4d} \sum_{k=0}^d \sum_{i=0}^k \binom{k}{i}^2 \\
    &= \epsilon^2 d \Big(\frac{2}{1-\delta}\Big)^{2d} 2^{4d} \sum_{k=0}^d \binom{2k}{k} \\
    &\leq \epsilon^2 d \Big(\frac{2}{1-\delta}\Big)^{2d} 2^{4d} \sum_{k=0}^d 4^k \\
    &\leq \epsilon^2 d \Big(\frac{2}{1-\delta}\Big)^{2d} 2^{4d} \frac{4^{d+1}-1}{3} \\
    &\leq \frac{\epsilon^2 d}{(1-\delta)^{2d}} 2^{8d}  \\
\end{split}
\end{equation*}
where the second equality and inequality follows from the well-known identities $\sum_{i=0}^k \binom{k}{i}^2 = \binom{2k}{k}$ and $\binom{2k}{k} \leq 4^k$. From the assumption that $\delta \leq 1/2$ this then implies the upper bound
\begin{equation*}
     \|p\|_2^2 \leq  \epsilon^2 d 2^{10d}.
\end{equation*}
This proves the claimed upper bound. For the lower bound, an application of the Cauchy-Schwartz inequality implies
\begin{equation*}
    |p(x)|^2 = |a_0 + a_1 x + ... + a_d x^d|^2 \leq (a_0^2 + a_1^2 + ... + a_d^2) (1 + x^2 + x^4 + ... + x^{2d}).
\end{equation*}
When $x=-1$ this reads
\begin{equation}
\label{eqn:2norm_bound}
    \|p\|_2^2 \geq \frac{1}{d+1}\Big|T_d\Big(\frac{2}{1-\delta}\Big)/T_d\Big(\frac{1}{1-\delta}\Big)\Big|^2.
\end{equation}
To prove the claimed lower bound, it suffices to show that $T_d(2x)/T_d(x) \geq (3/2)^d$ for all $x \geq 1$. We prove this inequality by induction on $d$. Since $T_0(x)=1$ and $T_1(x)=x$, it is easily seen that $T_0(2x)/T_0(x)=1$ and $T_1(2x)/T_1(x)= 2$, so that the cases $d=0$ and $d=1$ both hold. Now suppose the statement is true for all degrees at most $d$. Applying the recurrence of (\ref{eqn:chebyrecurrence}) we obtain
\begin{equation*}
\begin{split}
    \frac{T_{d+1}(2x)}{T_{d+1}(x)} &= \frac{4xT_d(2x)-T_{d-1}(2x)}{2xT_d(x)-T_{d-1}(x)} \\
    &\geq \frac{4xT_d(2x)-T_{d-1}(2x)}{2xT_d(x)}. \\
\end{split}
\end{equation*}
Since the Chebyshev polynomial takes values between between $-1$ and $1$ on the interval $[-1,1]$ then applying Lemma~\ref{lem-cheby-max} we obtain $T_d(2x) \geq T_{d-1}(2x)$ for all $x \in [1, \infty)$, thus 
\begin{equation*}
    \frac{T_{d+1}(2x)}{T_{d+1}(x)} \geq \frac{(4x-1)T_d(2x)}{2xT_d(x)}.
\end{equation*}
The function $\frac{4x-1}{2x}$ takes minimum value $3/2$ on the interval $[1, \infty)$, so applying this along with our inductive hypothesis gives $T_{d+1}(2x)/T_{d+1}(x) \geq (3/2)^{d+1}$, which then implies the claimed lower bound.
\end{proof}

\subsection{Proof of Lemma~\ref{lem-coef_bound2}}\label{app:L5.5}

\begin{lemma}[Restatement of Lemma~\ref{lem-coef_bound2}]
Let $\epsilon > 0$ and $\delta \in (0, 1/2]$ be given, and let $p_d(x)$ denote the polynomial 
\begin{equation*}
    p_d(x) = T_{d}\Big(\frac{x}{1-\delta}\Big)/T_d\Big(\frac{1}{1-\delta}\Big)
\end{equation*}
For $d \geq \frac{1}{\sqrt{\delta}} \log(4/\epsilon)$, the polynomial $p_d(x)$ has $2$-norm which satisfies
\begin{equation*}
  \|p_d\|_2 \leq \epsilon2^{3d-1}.
\end{equation*}
\end{lemma}

\begin{proof}
We proceed by induction on the degree $d$. Note that in general since the term $1/T_d(\frac{1}{1-\delta})$ is a constant then the norm of $p_d$ can be expressed as
\begin{equation*}
    \|p_d\|_2 = \frac{1}{|T_d(\frac{1}{1-\delta})|} \Big\| T_d \Big(\frac{x}{1-\delta} \Big) \Big\|.
\end{equation*}
The choice of $d \geq \frac{1}{\sqrt{\delta}} \log(4/\epsilon)$ implies that $\frac{1}{T_d(\frac{1}{1-\delta})} \leq \frac{\epsilon}{2}$, therefore in general will show 
\begin{equation*}
    \Big\| T_d \Big(\frac{x}{1-\delta} \Big) \Big\| \leq 2^{3d}.
\end{equation*}
For the cases $d=0,1$, the Chebyshev polynomials $T_d(x)$ are $1, x$, respectively. Therefore $T_d(\frac{x}{1-\delta})$ in these cases are $1, \frac{x}{1-\delta}$, respectively. The norms of these polynomials are then $1$ and $\frac{1}{1-\delta} \leq 2$,  since $\delta$ is assumed to be in the interval $(0 , 1/2]$. Therefore $\|T_d(\frac{x}{1-\delta})\| \leq 2^{3d}$ for $d=0$ and $1$. 

For the general case when $d\geq 2$, we assume that $\|T_i(\frac{x}{1-\delta})\| \leq 2^{3i}$ for all $i < d$. The Chebyshev polynomials satisfy the recurrence $T_d(x) = 2xT_{d-1}(x) - T_{d-2}(x)$ for $d\geq 2$. Since the polynomial norm is subadditive and since $\delta \in (0, 1/2]$ then 
\begin{equation*}
\begin{split}
    \Big\|T_d\Big(\frac{x}{1-\delta}\Big)\Big\| &\leq \Big\|\frac{2x}{1-\delta}T_{d-1}\Big(\frac{x}{1-\delta}\Big) \Big\| + \Big\|T_{d-2}\Big(\frac{x}{1-\delta}\Big) \Big\| \\
    &\leq 4\times 2^{3(d-1)} + 2^{3(d-2)} \\
    &\leq  2^{3d} \\
\end{split}
\end{equation*}
as claimed.
\end{proof}

\end{document}